\documentclass[12pt]{article}

\newcommand{\dtcolornote}[3][]{}

\providecommand{\bd}{broker-dealer\xspace}
\providecommand{\smc}{smart-contract\xspace}

\providecommand{\idt}{Inter-dealer trade\xspace}
\providecommand{\fl}{first-leg\xspace}
\providecommand{\sel}{second-leg\xspace}
\providecommand{\jpm}{joint profit maximization\xspace} 

\providecommand{\bs}{balance-sheet\xspace}
\providecommand{\g}{$T$\xspace}
\providecommand{\m}{$M$\xspace}
\providecommand{\mm}{$MM$\xspace}
\providecommand{\rmm}{$RM$\xspace}

\usepackage[top=1in, bottom=1in, left=1in, right=1in]{geometry}
\usepackage{amsmath}
\usepackage{amssymb}

\usepackage{comment}

\usepackage[english]{babel}
\usepackage{amsthm}
\usepackage{xcolor}
\usepackage{xspace}
\usepackage{subcaption}

\usepackage{csquotes}

\usepackage{float}
\usepackage{parskip}% http://ctan.org/pkg/parskip
\usepackage{adjustbox}
\usepackage{marvosym,graphicx,
endnotes,setspace}
\definecolor{rune}{HTML}{4A6672}

\usepackage{setspace}

\usepackage{tabularx}
\usepackage{url}

\usepackage[T1]{fontenc}
\usepackage{textcomp}
\usepackage{longtable}
\usepackage{booktabs, caption, makecell}

\usepackage{threeparttable}
\DeclareMathAlphabet{\mathpzc}{OT1}{pzc}{f}{it}
\usepackage{appendix}

\theoremstyle{plain}
\usepackage{graphicx}

\usepackage{parskip}% http://ctan.org/pkg/parskip
\usepackage{adjustbox}
\usepackage{multirow}
\usepackage{changepage}
\usepackage{titlesec}
\usepackage{tikz}
\usepackage{subcaption}
\definecolor{rune}{HTML}{4A6672}
\newtheorem{proposition}{Proposition}
\newtheorem*{proposition*}{Proposition 1}
\newtheorem*{proposition**}{Proposition 2}
\newtheorem*{proposition***}{Proposition 3}
\newtheorem*{proposition****}{Proposition 4}
\newtheorem*{proposition*****}{Proposition 5}
\newtheorem*{proposition******}{Proposition 6}
\newtheorem*{proposition*******}{Proposition 7}

\newtheorem{assumption}{Assumption}
\newtheorem{theorem}{Theorem}
\newtheorem*{corollary*}{Corollary 1}
\newtheorem*{corollary**}{Corollary 2}
\newtheorem*{corollary***}{Corollary 3}

\newtheorem{lemma}[theorem]{Lemma}

\newtheorem{definition}{Definition}
\newtheorem*{theorem*}{Theorem 1}
\newtheorem*{theorem**}{Theorem 2}
\newtheorem*{theorem***}{Theorem 3}
\newtheorem*{theorem****}{Theorem 4}
\newtheorem*{theorem*****}{Theorem 5}

\usepackage[many]{tcolorbox}   
\definecolor{main}{HTML}{000000}    % setting main color to be used
\definecolor{sub}{HTML}{FFFFFF}     % setting sub color to be used

\newtcolorbox{boxH}{
    colback = sub, 
    colframe = main, 
    boxrule = 0pt, 
    leftrule = 6pt % left rule weight
}

\usepackage{algorithm}
\usepackage[noend]{algpseudocode}
\usepackage{enumitem}
\usepackage{caption}
\providecommand{\keywords}[1]
{
  \small	
  \textbf{\textit{Keywords---}} #1
}

\usepackage{natbib}

\let\cite\citeyear

\title{A Smart-Contract to Resolve Multiple Equilibrium in an Intermediated Repo Trade}
\author{Daniel Aronoff\thanks{Research Affiliate, MIT Department of Economics} and Robert Townsend\thanks{Elizabeth \& James Killian Professor of Economics, MIT Department of Economics} }

\date{April 18, 2026}

\providecommand{\bd}{broker-dealer\xspace}
\providecommand{\smc}{smart-contract\xspace}

\providecommand{\idt}{Inter-dealer trade\xspace}
\providecommand{\fl}{first-leg\xspace}
\providecommand{\sel}{second-leg\xspace}
\providecommand{\jpm}{joint profit maximization\xspace} 

\providecommand{\bs}{balance-sheet\xspace}
\providecommand{\g}{$T$\xspace}
\providecommand{\m}{$M$\xspace}
\providecommand{\mm}{$MM$\xspace}
\providecommand{\rmm}{$RM$\xspace}

\usepackage[top=1in, bottom=1in, left=1in, right=1in]{geometry}
\usepackage{amsmath}
\usepackage{amssymb}

\usepackage{comment}

\usepackage[english]{babel}
\usepackage{amsthm}
\usepackage{xcolor}
\usepackage{xspace}
\usepackage{subcaption}

\usepackage{csquotes}

\usepackage{float}
\usepackage{parskip}% http://ctan.org/pkg/parskip
\usepackage{adjustbox}
\definecolor{rune}{HTML}{4A6672}

\usepackage{setspace}

\usepackage{tabularx}
\usepackage{url}

\usepackage{longtable}
\usepackage{booktabs, caption, makecell}

\usepackage{threeparttable}
\DeclareMathOperator*{\argmax}{argmax}
\usepackage{appendix}
\theoremstyle{plain}
\usepackage{graphicx}

\usepackage{parskip}% http://ctan.org/pkg/parskip
\usepackage{adjustbox}
\usepackage{multirow}
\usepackage{changepage}
\usepackage{titlesec}
\usepackage{tikz}
\usepackage{subcaption}
\definecolor{rune}{HTML}{4A6672}
\usepackage[many]{tcolorbox}   
\definecolor{main}{HTML}{000000}    % setting main color to be used
\definecolor{sub}{HTML}{FFFFFF}     % setting sub color to be used

\usepackage{algorithm}
\usepackage[noend]{algpseudocode}
\usepackage{enumitem}
\usepackage{caption}
\usepackage{geometry}

\begin{document}

\maketitle
\thispagestyle{empty}
\begin{abstract}
% === INLINED: Abstract-2.tex ===

We construct an empirically founded model of a repo trade intermediated by two broker-dealers and prove multiple equilibrium and the existence  of equilibrium at the joint profit maximizing volume of trade. We then present a smart contract that resolves multiple equilibrium by requiring each broker-dealer to report its client schedule and its minimum hurdle spread, and implementing a selection rule that filters out hurdle-infeasible outcomes. Whenever there exists an equilibrium that exceeds both hurdle spreads, the protocol selects the joint profit maximizing feasible trade and thereby avoids a collapse to no trade. The smart contract is a machine executed algorithm which eliminates the need for trust. Hardware and cryptography are used to prevent leakage of broker-dealer client trade schedules, and to enable privacy-protected auditing with zero-knowledge proofs of the integrity of computations. The outcome can be implemented by a myopic strategy where a broker-dealer truthfully reports its own variables without anticipating its counterparty’s reports. This minimizes cognitive and computational complexity, thereby making our smart contract suitable for real-world deployment.

% === END: Abstract-2.tex ===

\end{abstract}

% \begingroup\renewcommand\thefootnote{}\footnotetext{\par \vskip 1em \noindent We wish to acknowledge Glenn Ellison for his valuable comments on an earlier version of this paper (Aronoff \cite{aronoff2022thesis}). All errors are our responsibility.}

\keywords{smart contracts, multiple equilibrium, privacy, market design, coordination}

\bigskip
\noindent\textit{JEL:} D47, D82, G23, L86, C72.\\
\noindent\textit{ACM CCS:} Security and privacy $\rightarrow$ Distributed systems security; Theory of computation $\rightarrow$ Algorithmic mechanism design; Applied computing $\rightarrow$ Electronic commerce; Theory of computation $\rightarrow$ Algorithmic game theory; Computer systems organization $\rightarrow$ Distributed architectures.\\
\noindent\textit{MSC 2020:} 91A80; 91B26; 68M14; 94A60; 91B54.

\pagebreak

\tableofcontents
\pagebreak

\newpage
\pagenumbering{arabic}
\onehalfspacing

% === INLINED: Introduction-2.tex ===
\section{Introduction}
\label{sec:Introduction}

Unreliable information disclosure, limited commitment, and coordination failures are persistent obstacles to exchange in intermediated markets. In textbook mechanism design, these obstacles can often be mitigated by a central planner that collects messages, selectively discloses information, and assigns allocations. In practice, however, market participants are reluctant to delegate sensitive trading information to a trusted planner, and they are concerned about how that planner might use, leak, or be compelled to reveal the information. This paper shows how a coordination problem in an intermediated trade can be resolved by a self-executing mechanism that plays the role of the planner while minimizing the need for trust and preserving privacy.

We study a chain of exchange in which trade between the ultimate counterparties is intermediated by two broker-dealers. The motivating application is a repo transaction in the U.S.\ Treasury market. A repo is an exchange of cash for a security today with an agreement to reverse the exchange at a later date at a pre-specified repurchase price. In the environment we model, a cash lender and a cash borrower (the clients) each transact with their own broker-dealer, and the two broker-dealers transact with each other in the interdealer market. Each broker-dealer sets client-facing terms and an interdealer term, subject to feasibility and to internal constraints that reflect balance-sheet costs. The resulting trade volume is pinned down by the intersection of the clients' schedules and the interdealer terms that transmit those schedules across the chain.

The economic friction is that this transmission can admit multiple self-consistent outcomes. We show that the intermediated repo game can have a continuum of equilibria: distinct combinations of interdealer and client-facing terms that clear the chain and satisfy best responses. Multiplicity persists both when each broker-dealer observes the counterparty's client schedule and when such information is not observed (Theorems 2 and 5). A practical consequence is that equilibrium selection can be economically consequential. In particular, when broker-dealers face binding balance-sheet costs, they require a minimum return on the spread between what they pay one client and charge the other. A low-spread equilibrium may fail to satisfy this hurdle, in which case the trade is effectively curtailed despite the existence of other equilibria that exceed the hurdle spread of both \bd{s}. 

Post-financial crisis bank leverage regulation, such as the Supplementary Leverage Ratio, and  FASB accounting reforms increase the leverage impact of repo intermediation raise the shadow cost of balance-sheet usage and increase the minimum  spreads at which intermediation is profitable. Multiple equilibrium adds equilibrium-selection risk: when balance sheet must be allocated ex ante, dealers cannot condition capacity on which equilibrium will be selected.  Consequently, the ex-ante risk of infeasible equilibrium may cause banks to reduce capital allocations to their affiliate \bd{s}. We characterize the regulator’s objective as ensuring that trade occurs whenever there is a feasible equilibrium. Our contribution is to isolate and remove the equilibrium- selection component of the capacity shortfall.

To address the coordination problem created by multiplicity, we propose a smart-contract protocol that selects a particular allocation from the equilibrium set. The protocol receives from each broker-dealer (i) its client schedule and (ii) its minimum acceptable spread, and it computes the constrained joint-profit-maximizing trade over the set of hurdle-feasible allocations. The protocol does not generally maximize traded volume; instead, it commits the market to a feasible equilibrium-selection rule that selects the \bd \jpm volume whenever the hurdle-feasible set is nonempty, thereby ruling out a collapse to no trade driven by equilibrium-selection uncertainty. We provide an implementation result under minimal assumptions: truthful reporting of schedules and minimum spreads by \bd{s} is a myopic best response that constitutes a subgame-perfect equilibrium of the smart-contract game (Theorem 3). The mechanism can be used without requiring broker-dealers to form beliefs about how the counterparty will report. The protocol therefore resolves equilibrium selection while also reducing cognitive and strategic complexity for participants.

This paper makes three contributions. First, we develop an empirically disciplined model of intermediated repo trade in which two client-facing broker-dealers interact through an interdealer market, and we characterize how the chain structure can generate multiple equilibria. Second, we construct and characterize a constrained joint-profit-maximizing allocation and show that it can be implemented through a simple reporting game that selects a unique outcome from the equilibrium set, both under full information and under asymmetric information about clients' schedules. Third, we provide a trust-minimizing implementation blueprint: the selection rule can be executed as code, and the reporting and computation can be structured to protect the privacy of client schedules and to support verification of correct execution through privacy-preserving auditing technologies. In combination, these results connect a standard equilibrium-selection problem in intermediated markets to a concrete mechanism that is implementable in institutional settings.

The practical significance of the mechanism is tied to the post-crisis scarcity of dealer balance-sheet capacity in the U.S.\ Treasury and repo markets. Repo intermediation consumes balance sheet and is traditionally a low-margin matched-book activity. Since the 2008 financial crisis, regulatory reforms such as the Supplementary Leverage Ratio have raised the shadow cost of balance-sheet usage and strengthened incentives to ration internal capital toward higher-yield activities. In stressed episodes, constrained dealer capacity can impair liquidity provision and disrupt the transmission of monetary policy, as illustrated by the Treasury market turmoil of March 2020. In such environments, the difference between equilibria is not merely a theoretical curiosity: equilibrium selection can determine whether intermediated trade occurs at all when minimum-return hurdles bind.

From a regulatory standpoint, the first-best outcome would maximize intermediated repo volume subject to dealer participation under leverage constraints. Our protocol is a second-best infrastructure change: it does not relax leverage regulation or subsidize balance-sheet usage, but it increases effective intermediation capacity by eliminating equilibrium-selection risk and guaranteeing feasibility whenever feasible intermediation exists.\footnote{Notably, our \smc does not guarantee the selection of the maximin volume of feasible trade. This would be the preferred regulatory objective, but we have not found an incentive mechanism to implement it.} By coordinating broker-dealers on the constrained joint-profit-maximizing allocation, the proposed protocol increases the profitability of intermediation, which can make repo activity a more attractive use of scarce internal capital and thereby expand intermediation capacity at the margin.

\paragraph{Related literature.}
This paper builds on three related strands of literature that jointly inform our model of intermediated repo trade and the smart-contract implementation. The first concerns theoretical models of repo market intermediation. Duffie \cite{Duffie1996}, Huh and Infante \cite{HuhInfante2020}  and Gottardi et.al. \cite{GottardiMonnet2017} show how dealer balance-sheet constraints, collateral scarcity, and rehypothecation collectively determine repo rates and market liquidity. Their frameworks provide a  microfoundation for the interdealer structure and hurdle-rate constraint which are resonant with our model. The second strand addresses coordination failures and multiplicity of equilibria in financial markets. Morris and Shin \cite{MorrisShin1998}, Martin, Skeie, and von Thadden \cite{MartinSkeieThadden2014}, and Cole and Kehoe \cite{ColeKehoe2000} build models that resolve self-fulfilling runs or multiple equilibria through information structure or institutional design. We design a \smc that resolves multiple equilibrium and selects the equilibrium while being agnostic toward the information structure.  The final strand examines  how \smc{s} and blockchain technologies can implement financial market coordination. Cong and He \cite{CongHe2019}, Chiu and Koeppl \cite{ChiuKoeppl2019} and Townsend and Zhang \cite{Townsend2023} demonstrate that programmable contracts can replicate trust and enforcement functions traditionally performed by intermediaries, offering a theoretical basis for our algorithmic mechanism that selects the joint profit–maximizing equilibrium while preserving privacy and auditability.

\paragraph{Roadmap}
The remainder of the paper is organized as follows. Section~2 collects the empirical and theoretical elements that motivate the model. Section~3 presents the intermediated repo game and the broker-dealers' objectives. Section~4 defines equilibrium and establishes existence and multiplicity. Section~5 characterizes joint profit maximization and the target allocation. Section~6 extends the analysis to asymmetric information. Section~7 introduces balance-sheet costs and the minimum-spread restriction. Section~8 presents the smart-contract protocol and proves the implementation results, including integrity and privacy-preserving auditability. Section~9 concludes.

\section{Elements of The Model}
\label{sec:Elements_of_the_model}

\subsection{Motivation}

In an over-the-counter (OTC) and other financial markets, trade between buyers and sellers is intermediated by broker-dealers\footnote{In an over-the-counter (``OTC'') market agents engage in bilateral exchanges, as opposed to a centralized exchange where e.g. agents submit limit orders from which a single market clearing price is determined.}. Intermediaries trade with their buyer/seller clients and with each other. The largest trading market in the world, the US Treasuries repo market is intermediated. A repo trade involves two contracts entered into simultaneously whereby one counterparty initially purchases the financial asset (the ``first-leg'') and commits to repurchase the same financial asset at a later date (the ``second-leg'').\footnote{In a repo trade the volume of financial asset is the same at both legs. The money exchanged can differ.} The estimated daily volume of US Treasuries repo is \$12 Trillion (Hempel et.al. \cite{HempelKahnShephard2025}). The market is intermediated by broker-dealers who conduit Treasuries from repo lenders - typically hedge funds, insurance companies and pension funds, to repo lenders, typically money market funds (Kahn et.al. \cite{Kahn2021}).

The revealed preference of this form of market organization suggest it is useful. However, intermediated trade creates a potential coordination problem which can result in inefficiently low trading due to the presence of multiple equilibrium. Multiple equilibrium is caused by a negative externality from the client transaction volume entered into by broker-dealers, which determines the volume of trade they will desire to offset in the interdealer market (e.g. if a broker-dealer purchases $x$ units of the financial asset from its client, it will desire to sell $x$ units to another broker-dealer). A low volume outcome can occur when a broker-dealer anticipates that its interdealer counterparty will arrive with a low volume of money or financial asset to trade and reduces the volume of trade with its client. In reaction the counterparty adjusts downward its client trade volume and so forth, since neither one wishes to carry excess inventory. This can lead to a self-reinforcing low-trade equilibrium.

\subsection{Model elements}

In this section we state the key elements of our model and point to empirical theoretical studies that motivate and support our design choices.

\subsection*{A repo trade}
\label{subsec:repo_trade}

\begin{boxH}
\textbf{Elements}: 

A repo trade consists of 2 transactions involving the sale of a financial asset (which we denote $T$), entered into simultaneously,  that occur sequentially in time, the first-leg and the second-leg.
\[r = (p_{2} - p_{1})/p_{1}\]
$r$ is the repo rate, $p_{1}$ is unit the price of the financial asset at the first-leg and $p_{2}$ is the unit price at the second-leg. We normalize $p_{1}=1$.
\end{boxH}
When excess collateral is provided by the first-leg seller the financial asset can be expressed as $T = T_{asset} + h$, where $h$ is the ``haircut''. The decomposition is inessential to our model. 

\subsection*{Repo chains}
\label{subsec:Repo chains}

\begin{boxH}
\textbf{Elements}: 

2 Clients; ultimate repo lender \mm and ultimate repo borrower \rmm. 2 intermediary \bd{s}, $BD_{mm}$ and $BD_{rm}$. 

2 Traded objects; money \m and Treasuries securities \g.

2 Trading dates and contracts. First-leg, borrower sells \g to lender today. Second-leg, borrower repurchases \g from lender at a future date.
\end{boxH}

The setup is a chain of repo transactions composed as follows. There are 2 clients, a repo borrower who sells the financial asset, $T$, in exchange for money, $M$, at the first-leg and repurchases the at a later time called the second-leg. The repo borrower's supply function for the financial asset is $T = S_{r}$ and the repo lender's demand function for the financial asset is $T = D_{r}$.  There are 2 broker-dealers. Each client trades with a broker-dealer and each broker-dealer attempts to re-trade the financial object it acquired from its client with the counterparty broker-dealer (the "interdealer" trade); e.g at the first-leg the repo borrower's broker-dealer purchases the financial asset from its client and sells it to the counterparty broker-dealer. At the second-leg, that broker-dealer re-purchases the financial asset from the counterparty broker-dealer and re-sells it to its client. Figure \ref{fig:FirstSecond-Leg-Repo-Chain} depicts the repo chain formed by these trades. At the \fl $T$ flows left to right, from the repo borrower $RM$ - typically a hedge fund or a ``risk manager'' -  to repo lender $MM$ - typically a money market fund or ``cash manager'' (Pozsar \cite{Pozsar2014}). In between are the broker-dealers, and $M$ flows in the opposite direction. The flows reverse at the \sel.

\begin{figure}[H]
\center
{\footnotesize{
$MM \underset{\text{$BD_{h}$ sale of $T$ to $MM$}}{\underbrace{
   \begin{aligned}
        T & \longleftarrow\\
        \longrightarrow & M
    \end{aligned}
}} BD_{mm}\underset{\text{interdealer sale of $T$}}{\underbrace{
    \begin{aligned}
        T & \longleftarrow\\
        \longrightarrow & M
    \end{aligned}
}} BD_{rm}\underset{\text{$MM$ sale of $T$ to $BD_{j}$}}{\underbrace{
    \begin{aligned}
        T & \longleftarrow\\
        \longrightarrow & M
    \end{aligned}
    }} RM$
}}
\end{figure}

\begin{figure}[H]
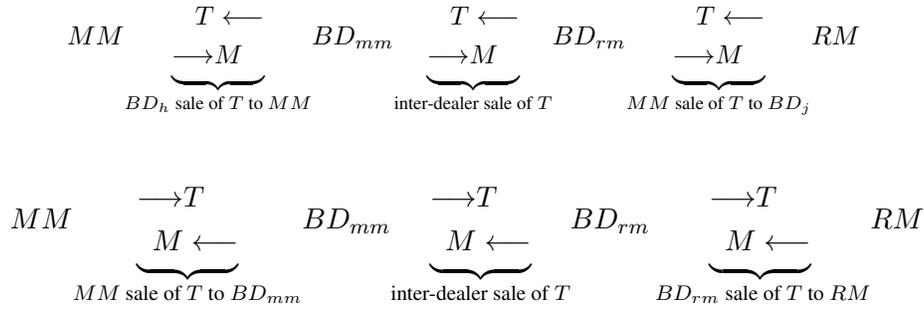

\center
{\small{
$MM \underset{\text{$MM$ sale of $T$ to $BD_{mm}$}}{\underbrace{
   \begin{aligned}
        \longrightarrow & T\\
        M &\longleftarrow
    \end{aligned}
}} BD_{mm}\underset{\text{interdealer sale of $T$}}{\underbrace{
    \begin{aligned}
        \longrightarrow & T\\
        M & \longleftarrow
    \end{aligned}
}}BD_{rm}\underset{\text{ $BD_{rm}$ sale of $T$ to $RM$}}{\underbrace{
    \begin{aligned}
        \longrightarrow & T\\
        M & \longleftarrow
    \end{aligned}
    }} RM$
\caption{First and Second-Leg Repo Chain}
\label{fig:FirstSecond-Leg-Repo-Chain}
}}
\end{figure}

Aronoff et.al. \cite{Aronoff2025OFR} decompose the US Treasuries repo network into chains and cycles of repo trades. The most prevalent configuration is the chain depicted in Figure \ref{fig:FirstSecond-Leg-Repo-Chain}.

\subsection*{Timing of repo trades}

\begin{boxH}
\textbf{Elements}:

Contracts for \fl and \sel transactions are entered into simultaneously at the \fl date.

Contracts between \bd{s} and its clients are committed to before interdealer contracts between \bd{s} are committed to.
\end{boxH}

Broker-dealers typically commit first to client-facing transactions and  subsequently offset these exposures through interdealer trades. Gardner and Huh \cite{GardnerHuh2024} document that ``DC-ID trades are instances in which customer trades are prearranged with offsetting interdealer trades,'' providing direct evidence that dealers arrange interdealer commitments only after securing client orders. Similarly, Huh and Infante \cite{HuhInfante2020} emphasize that ``after dealers receive client orders, dealers transact in the interdealer cash and repo market,'' underscoring the sequencing of commitments across markets. At the microstructure level, Macchiavelli and Pettit \cite{MacchiavelliPettit2019} explain that a broker-dealer that extends a loan to a hedge fund client ``typically rehypothecates some of these securities received as collateral in the repo markets,'' highlighting the practice of funding client trades ex post through wholesale interdealer borrowing. Consistent with these accounts, Hempel, Kahn, and Shephard \cite{HempelKahnShephard2025} show that large dealer-banks operate matched books with ``nearly offsetting repo and reverse repo positions,'' which arise precisely because client transactions are subsequently hedged through interdealer trades. Together, this evidence establishes that the sequence of commitments runs from client to dealer, then from dealer to interdealer, a timing structure that has important implications for market liquidity and risk transmission.

We model trading activity as taking place in two sequential steps. (1) clients and $BD$'s agree to trades. (2) after that $BD$'s trade together in a simultaneous move game. Figure \ref{fig:Repo-Bilateral-Timing-Current} displays the repo market segments and the timing of repo trades in our model. Broker-dealer intermediaries trade with each other in the interdealer market and with their clients in the other two markets (Kahn and Olsen \cite{Kahn2021}). Broker-dealers typically lend to hedge-funds in the bilateral uncleared market and borrow from  money market funds in the Tri-party market (Hempel et.al. \cite{Hempel2023} and Paddrik et.al. \cite{paddrik2021}).

\begin{figure}[H]
\begin{center}
\includegraphics[page=1,width=0.64\textwidth,height = 0.26
\textheight]{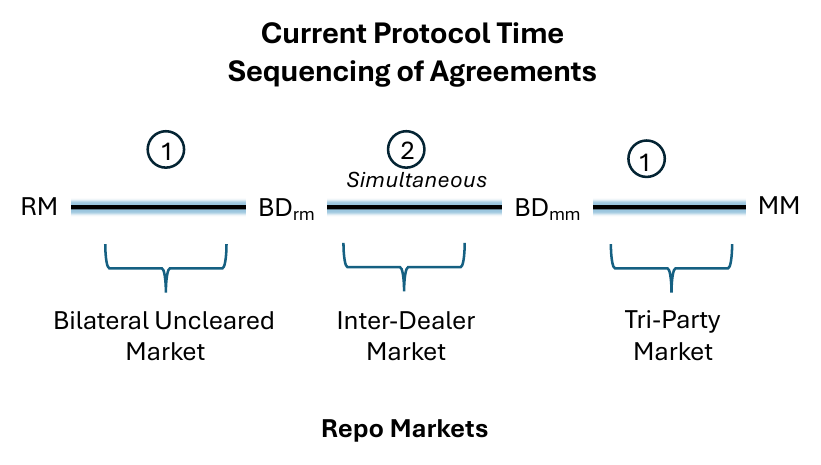}
\caption{\textit{\textbf{Timing under current protocol}. Sequential ordering by number. First, client trades are agreed to. Second, the interdealer trade is agreed to.}}
\label{fig:Repo-Bilateral-Timing-Current}
\end{center}
\end{figure}

\subsection*{Broker-dealer price-setting in the client market}

\begin{boxH}
\textbf{Elements}:
Broker-dealers are monopolist price-setters in the client markets.
\end{boxH}

We model the \bd as a price-setting monopolist in the client market, consistent with extensive empirical and theoretical research (Infante \cite{Infante2015Liquidity}, Avalos et al. \cite{Avalos2019RepoStress}, Eisenschmidt et al. \cite{Eisenschmidt2022SegmentedMarkets}, Duffie et al. \cite{Duffie2023DealerCapacity}, and Hempel et al. \cite{Hempel2024RepoIntermediation}). These studies provide evidence that broker-dealers often extract informational or relationship-based rents in client-facing trades. For example, Infante \cite{Infante2015Liquidity} models the dealer as a monopolist, illustrating how dealer pricing power can distort client outcomes. Avalos et al. \cite{Avalos2019RepoStress} note that U.S. repo markets currently ''rely heavily on [\bd affiliates of] four banks as marginal lenders'', implying that cash providers face limited counterparty options. Similarly, Eisenschmidt et al. \cite{Eisenschmidt2022SegmentedMarkets} find that ''dealers have market power over their customers, allowing them to charge OTC repo rates that differ from the interdealer rate,'' pointing to price discrimination across segments. A contributing factor enabling such behavior is the persistence of trading relationships: repeated interactions with hedge funds (Kruttli et al. \cite{Kruttli2022}, Financial Stability Board \cite{FSB2023}) and money market funds (Han et al. \cite{HAN2022}, Hempel et al. \cite{Hempel2023-MMF}) facilitate dealer market power and enable the exercise of monopolistic pricing in the client market.

\subsection*{Concavity of  money market fund demand and hedge fund supply}

\begin{boxH}
\textbf{Elements}:

The \mm demand function (for acquiring \g at the \fl) is denoted $D_{r_{mm},\theta_{mm}}$ where $r$ is the repo rate and $\theta_{mm}$ is a stochastic shock.

$D$ is concave increasing - and twice differentiable - in $r$.

$\theta_{mm}$ shifts $D$ monotonically.

The \rmm supply function (for selling \g at the \fl) is denoted $S_{r_{rm},\theta_{rm}}$ where $r$ is the repo rate and $\theta_{rm}$ is a stochastic shock.

$S$ is concave decreasing - and twice differentiable - in $r$.

$\theta_{rm}$ shifts $S$ monotonically.
\end{boxH}

A body of empirical and theoretical work supports the view that both hedge fund borrowing demand and money market fund (MMF) repo supply are concave in relevant regions. On the borrower side, He, Nagel, and Song \cite{HeNagelSong2022JFE} show that ``a simultaneous rise in reverse-repo rates charged by dealers reduces the attractiveness of a levered investment in Treasuries for (risk-averse) hedge funds,'' emphasizing that incremental demand falls disproportionately as funding costs increase. Relatedly, Banegas and Monin \cite{BanegasMonin2023FEDS} document that ``a 200 bp minimum haircut would reduce effective leverage from 56x to 25x,'' underscoring how small increases in margin requirements sharply reduce borrowing, consistent with a diminishing marginal willingness to lever. On the lender side, Huber \cite{Huber2023JFE} structurally estimates that MMFs exhibit ``aversion to portfolio concentration and preference for stable lending,'' leading to a supply schedule in which incremental willingness to lend to a given dealer diminishes with position size. Complementing this, Hempel and Monin \cite{HempelMonin2023FEDS} highlight that the ON RRP facility ``helps to set a floor on lending rates for MMFs,'' introducing a kink that flattens supply below the policy floor, while supply above the floor remains upward-sloping but with declining elasticity. Taken together, these findings provide both theoretical and empirical support for modeling hedge fund borrowing demand and MMF repo supply as concave functions.

\subsection*{Bargaining and surplus splitting in interdealer trades}

\begin{boxH}
\textbf{Elements}:

Broker-dealers split the surplus generated by their trade, where the surplus is the margin between the repo rate received from $RM$ and the repo rate paid to $MM$, multiplied by \g, the volume of trade between the \bd{s}.\footnote{$T$ is a monetary unit since the first-leg unit price is $p_{1} = 1$.}

\[Surplus\; = (r_{rm} - r_{mm})T\]
\end{boxH}

A literature models bilateral dealer-to-dealer transactions as the outcome of bargaining, with transaction prices reflecting a division of surplus rather than price-taking. In the canonical OTC search model of Duffie, G\^arleanu, and Pedersen \cite{DuffieGarleanuPedersen2005}, ``prices are set through a bargaining process that reflects each investor's or marketmaker's alternatives to immediate trade,'' and dealers ``capture part of the difference between the interdealer price and investors' reservation values.'' Building on this framework, Weill \cite{Weill2020NBER} shows that the generalized Nash bargaining solution yields prices that are ``a convex combination of the counterparty’s reservation value and the interdealer benchmark,'' making explicit the surplus-splitting logic of bilateral dealer trades. Empirical evidence corroborates these theoretical predictions: Hendershott, Li, Livdan, and Schürhoff \cite{HendershottLiLivdanSchurhoff2020} estimate that ``dealers' bargaining power on the buy and sell sides is large, at 98\% and 94\%,'' implying that dealers capture nearly all of the surplus in bilateral negotiations. Similarly, Feldhütter \cite{Feldhutter2012RFS} finds that dealers capture 97\% of trade surpluses in corporate bonds, while Pinter, Üslü, and Wijnandts \cite{PinterUsluWijnandts2025BIS} show that such bargaining dynamics extend to government bond markets as well. Taken together, these studies provide both theoretical and empirical support for modeling interdealer transactions as surplus-splitting bargains, in which the price between two broker-dealers divides the surplus generated by their respective client repo trades.

\subsection*{Inventory penalties}

Let $Q(r_{mm},r_{rm}) := \min\{D(r_{mm}),S(r_{rm})\}$ be the interdealer trade quantity and $r^{*} =$ the interest rate earned on unsold inventory (which we set = 0). Assume each broker-dealer incurs a penalty that depends only on its own unsold inventory:

\begin{boxH}
\textbf{Elements}:
\[gap(BD_{mm}) = G_{mm}\!\left(D(r_{mm})-Q(r_{mm},r_{rm})\right)\]
\[gap(BD_{rm}) = G_{rm}\!\left(S(r_{rm})-Q(r_{mm},r_{rm})\right),\]
where $G_{mm} = r_{mm} - r^{*}$; $G_{rm} = r_{rm} - r^{*}$, and 
$G_{mm},G_{rm}:\mathbb{R}_+\to\mathbb{R}_+$ satisfy $G_i(0)=0$ and are (weakly) increasing, and strictly increasing on $(0,\infty)$.
\end{boxH}

We model each broker--dealer as facing a cost of carrying an unmatched position when its client-side transaction cannot be fully offset in the interdealer market. This inventory-control motive is a core feature of dealer-market microstructure: Stoll  \cite{Stoll1978DealerServices} shows that dealers require compensation for providing immediacy while bearing inventory risk, and Ho et.al. \cite{HoStoll1981OptimalDealerPricing} formalize how optimal dealer quotes adjust endogenously to manage inventory under return uncertainty. When inventory must be financed, funding conditions interact with inventory exposure, so that tighter funding and higher margins raise the effective carrying cost and reduce liquidity provision (Brunnermeir and Pederson \cite{BrunnermeierPedersen2009MarketLiquidityFundingLiquidity}). In the specific context of repo intermediation, these inventory penalties arise even for matched-book dealers because client-facing repo trades are typically committed before offsetting interdealer transactions, exposing broker-dealers to interim cash or collateral imbalances that must be financed on balance sheet (Hemplel et.al. \cite{HempelKahnShephard2025}). We capture these effects in reduced form by allowing each broker--dealer to incur a penalty that is increasing in its own unsold inventory, scaled by the opportunity cost of funding relative to the outside return.

\subsection*{Minimum spread in repo balance-sheet allocation}
\label{subsec:Hurdlerates}

\begin{boxH}
\textbf{Elements}:

\[1/2(r_{rm} - r_{mm}) \geq \kappa\]
\end{boxH}

We assume that each broker--dealer requires a minimum spread between the repo rate charged to borrowers and the rate paid to lenders in order to allocate balance-sheet capacity to intermediation. This assumption is grounded in the post-crisis literature on balance-sheet constraints: Duffie \cite{Duffie1996} emphasizes that repo intermediation is a low-margin activity whose feasibility depends on the compensation dealers receive for deploying scarce balance sheet, while Anderson et.al. \cite{Andersen2019} show how funding and balance-sheet costs translate into required spreads even for low-risk, collateralized trades. Regulatory leverage constraints further raise these internal return thresholds: Duffie et.al. \cite{Duffie2023DealerCapacity} document that dealer balance-sheet scarcity materially limits Treasury and repo market intermediation, and Chabot et.al. \cite{Chabot2024} provide direct dealer-level evidence that intermediaries demand persistent positive spreads in repo markets as compensation for balance-sheet usage. Motivated by this evidence, we model a reduced-form hurdle rate that restricts intermediation to outcomes in which the equilibrium spread exceeds each broker--dealer’s internal minimum return requirement.

% === END: Elements_of_the_model.tex ===

% === INLINED: A_Model_of_Intermediated_Repo_Trade-2.tex ===

\section{A Model of Intermediated Repo Trade}
\label{sec:A Model of Intermediated Repo Trade}

In this section we model the trade between two \bd{s} depicted in Figure \ref{fig:FirstSecond-Leg-Repo-Chain}. $BD_{mm}$ has a client (or clients) $MM$ from which it has entered into a repo trade to borrow money in exchange for collateral. $BD_{rm}$ has a client (or clients) $RM$ from which it has entered into a reverse-repo trade to lend money in exchange for collateral. We use the model elements stated in Section \ref{sec:Elements_of_the_model} for trade chain structure, timing of trades, broker-dealer price-setting in the client market, concavity of client supply and demand functions and surplus-splitting between \bd{s}.\footnote{We impose a hurdle rate in Section \ref{subsec:minimum_spread}.}

\subsection{Broker-dealer objectives and information}

Our baseline model assumes common knowledge of client shocks. This enables us to suppress the $\theta$ variables. When each \bd observes the realization of its counterparty client's $\theta$, the analysis of equilibrium existence and multiplicity, which is our subject, uses only the concavity of client supply and demand functions. We relax the common knowledge restriction in Section \ref{sec:Equilibrium Under Asymmetric Information} and Appendix \ref{app:asymmetric_eqm}, where we show that the results we derive here for existence and multiplicity of equilibrium continue to hold when there is asymmetric information about counterparty client preferences and those preferences are subjected to shocks. In Appendix \ref{app:funding_liquidity} we relax the monopolistic price-setting assumption and allow \bd{s} to be market-makers who commit to trade minimum volumes with their clients.

\subsubsection*{The interdealer repo rate}
\label{subsec:The interdealer repo rate}

$BD_{mm}$ and $BD_{rm}$ negotiate a split of the intermediation surplus $(r_{rm} - r_{mm})T$. The surplus split is reflected in the interdealer repo rate $r_{bd}$. The share of the surplus allocated to $BD_{mm}$ is $(r_{bd} - r_{mm})T$, which is the spread between the rate at which $BD_{mm}$ borrows from $MM_{i_{i}}$ and lends to $BD_{rm}$. The share of the surplus allocated to $BD_{rm}$ is $(r_{rm} - r_{bd})T$, which is the spread between the rate at which $BD_{rm}$ borrows from $BD_{mm}$ and lends to $RM$. We are agnostic on the allocation of the surplus between $BD$'s or the bargaining process by which the allocation is determined. We posit an equal split for convenience, not necessity.\footnote{The only properties of $r_{bd}$ upon which our argument relies are that $\frac{\partial r_{bd}}{\partial r_{mm}} < 0$ and $\frac{\partial r_{bd}}{\partial r_{rm}} > 0$.}. The inter - dealer repo rate is
 
\begin{equation}
\label{eq:Inter - Dealer repo rate}
 r_{bd} = \frac{1}{2}(r_{rm} - r_{mm}) + r_{mm}  
 \end{equation}

\subsubsection*{Broker-dealer client pricing}
\label{subsec:brkoer-dealer-client-pricing}

Broker-dealers are price-setters in the client markets. For any volume of $T$ traded in the inter-dealer market, a \bd will set its client repo rate at the point the client demand or supply schedule intersects that volume. The reasoning is immediate. At each $T$ the $MM$ demand schedule lists the minimum repo rate at which it is willing to lend and the $RM$ supply schedule lists the maximum repo rate at which it is willing to borrow. From equation \ref{eq:Inter - Dealer repo rate} it can be seen that $BD_{mm}$ maximizes its profit by choosing the minimum $r_{mm}$ compatible with volume $T$ and $BD_{rm}$ maximizes its profit by choosing the maximum $r_{rm}$ compatible with volume $T$. 

\subsection{BD\textsubscript{mm}'s problem}
\label{subsec: $BD_{mm}$ problem_certainty}

$BD_{mm}$ sets its client repo rate $r_{mm}$ in the morning at $t_{1M}$. In order to form a rational expectation of the interdealer repo rate $r_{bd}$ and transaction volume that it will encounter afterwards, it must forecast $BD_{rm}$'s client repo rate (from which it can also forecast $BD_{rm}$'s client transaction volume) $r_{rm}$. equation \ref{eq: $BD_{mm}$ problem $r_{bd}$} is $BD_{mm}$'s problem, contingent on its forecast of $r_{rm}$.

\begin{equation}
\label{eq: $BD_{mm}$ problem $r_{bd}$}
\argmax_{r_{mm}}(r_{bd} - r_{mm})^{+}Q  - \underset{gap(BD_{mm})}{\underbrace{(D_{r_{mm}} - Q)^{+}r_{mm}}}
\end{equation}

Subject to $Q = \text{min}\{D_{r_{mm}}, S_{r_{rm}})\}$ i.e. the amount of the interdealer trade. $r_{mm} \in [0, r_{rm}]$ is the domain of $r_{mm}$\footnote{We assume $BD_{mm}$'s profit is non-negative, which requires $r_{mm}\leq r_{rm}$. We assume there is a lower bound repo rate at which $MM$'s demand for $T$ is zero. We normalize this rate to be $0$. This sets the bounds to the domain for $r_{mm}$.}.

We can substitute in equation \ref{eq:Inter - Dealer repo rate} and rewrite $BD_{ii}$'s problem as follows.

\begin{equation}
\label{eq:$BD_{mm}$ problem}
\argmax_{r_{mm}}\frac{1}{2}(r_{rm} - r_{mm})^{+}Q  - gap(BD_{mm})
\end{equation}

In its client market, $BD_{mm}$ conveys $T$ to $MM$ as collateral for a repo loan of $M$. $BD_{mm}$ carries the $M$ into the \idt. The first term of $BD_{mm}$'s problem is its profit from the interdealer repo trade. The second term, $gap(BD_{mm})$, represents $BD_{mm}$'s excess inventories; it is its loss incurred on borrowings of $M$ in its client market that $BD_{mm}$ is unable to trade with its $BD_{rm}$ counterparty.

To analyze $BD_{mm}$'s problem we partition it into two parts; (i) the case where $BD_{mm}$ carries excess inventories and (ii) the case where $BD_{mm}$ does not carry excess inventories. 

\subsubsection*{BD\textsubscript{mm} carries excess inventories of \textit{M} into the interdealer trade. gap\textsubscript{B,Q} \textgreater{}  0}
\label{subsec: $BD_{mm}$ carries excess inventories of $M$ into the interdealer trade. $gap(BD_{mm}) >0$}

Noting that $Q = S_{r_{rm}}$, $BD_{mm}$'s problem is

\begin{equation}
\argmax_{r_{mm}}\frac{1}{2}(r_{rm} - r_{mm})^{+}S_{r_{rm}} - gap(BD_{mm})
\end{equation}

The derivative of $BD_{mm}$'s problem is: $\underset{< 0}{\underbrace{-\frac{1}{2}S_{r_{mm}}}}\underset{< 0}{\underbrace{ - gap(BD_{mm})'}} < 0$\footnote{When $gap(BD_{mm})> 0$, then $\frac{\partial gap(BD_{mm})}{\partial r_{B,i}} := gap(BD_{mm})' = D'r_{mm} + (D_{r_{mm}} - Q) >0$, since $D^{'} > 0$.}. Therefore, $BD_{mm}$ will choose a corner solution, setting its client repo rate $r_{mm}$ to achieve $D_{r_{mm}} = S_{r_{rm}}$. This ensures $BD_{mm}$ will not carry any inventories it cannot sell into the \idt.

\subsubsection*{BD\textsubscript{mm} carries no excess inventories of \textit{M} into the \idt. gap\textsubscript{B,Q} = 0}

The domain of the problem is $r_{mm}\in [0, r_{rm}]$ such that $D_{r_{mm}} \leq S_{r_{rm}}$. It is a compact interval \footnote{Compactness is proved as follows. The derivative of $D_{r_{mm}}$ is strictly positive and continuous, therefore, by the Inverse Function Theorem, its inverse exists and is continuous (Rudin (1976) \cite{rudin} Theorem 9.24, p. 221.). The domain of the inverse function, $[0,S_{r_{rm}}]$, is compact. Therefore, the range of the inverse function, which is the domain of $BD_{mm}$'s problem, is a compact interval (Rudin (1976) \cite{rudin} Theorem 4.14, p. 89).}. $BD_{mm}$'s problem is

\begin{equation}
\label{eq: $BD_{mm}$ problem no excess inventory}
\argmax_{r_{mm}}\frac{1}{2}(r_{rm} - r_{mm})D_{r_{mm}} 
\end{equation}

$BD_{mm}$'s problem is a continuous concave function on the compact interval $r_{mm} \in [0, r_{rm}]$. The concavity reflects the tension between the value of increasing transaction amount $D_{r_{mm}}$ against the increased cost, $r_{mm}$, required to obtain it, which reduces $BD_{mm}$'s profit margin $\frac{1}{2}(r_{rm} - r_{mm})$\footnote{We simplify notation by setting $\frac{d D_{r_{mm}}}{d r_{mm}} := D^{'}$ and  $\frac{d^{2} D_{r_{mm}}}{(d r_{mm})^{2}} := D^{"}$. Similarly for $S_{r_{rm}}$.}.

\subsubsection*{Characterizing BD\textsubscript{mm}'s decision problem}

\begin{figure}[H]
\begin{center}
\includegraphics[page=1,width=0.55\textwidth,height = 0.25\textheight]{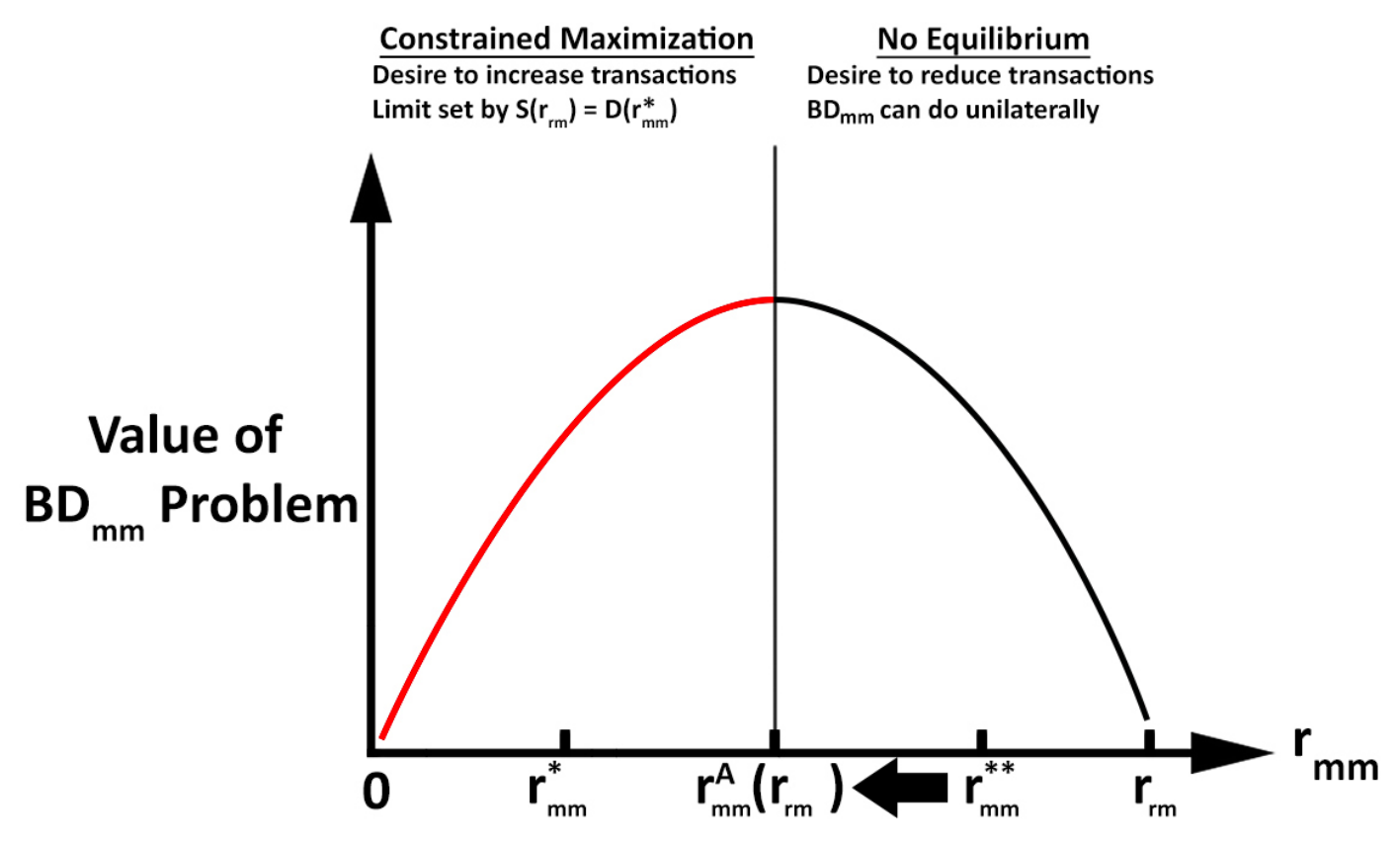}
\end{center}
\caption{$BD_{mm}$ Problem}
\label{fig:$BD_{mm}$ problem}
\end{figure}

Figure \ref{fig:$BD_{mm}$ problem} shows the curve defining the value of $BD_{mm}$'s problem at different repo rates for a given value of $r_{rm}$.

The value of $BD_{mm}$'s problem is zero at $r_{mm} = 0$ and at $r_{mm} = r_{rm}$. There are no transactions at the former and no marginal profit at the latter. The concavity of the problem ensures that it is single-peaked between these two endpoints. This peak is depicted in Figure \ref{fig:$BD_{mm}$ problem} as $r_{mm}^{A}(r_{rm})$. $BD_{mm}$ will choose the maximum repo rate $r_{mm}$ that (i) does not exceed $r_{mm}^{A}$ and (ii) does not cause $gap(BD_{mm}) > 0$. To see (i), note that $BD_{mm}$'s problem is negatively sloped to the right of $r_{mm}^{A}$ and it reduces inventory by reducing $r_{mm}$. Therefore, it will not reduce profit by reducing $r_{mm}$, since $gap_{r_{mm},Q}$ will shrink. (ii) reflects that $BD_{mm}$ will not carry excess inventory in equilibrium.

To gain further insight we evaluate $BD_{mm}$'s decision at repo rates on either side of the maximum. Suppose $r_{mm}^{*}$ in Figure \ref{fig:$BD_{mm}$ problem} is the repo rate at which $gap(BD_{mm}) = 0$. $BD_{mm}$ cannot increase $r_{mm}$, as doing so would cause $gap(BD_{rm}))$ to increase above zero. Since the value of $BD_{mm}$'s problem is increasing at $r_{mm}^{*}$, $BD_{mm}$ maximizes its profit by choosing $r_{mm}^{*}$, which we will call a constrained equilibrium repo rate for $BD_{mm}$\footnote{For $BD_{mm}$, any repo rate weakly below $r_{mm}^{A}(r_{rm})$ is a possible constrained outcome repo rate.}. By contrast, at $r_{mm}^{**}$ in Figure \ref{fig:$BD_{mm}$ problem} $BD_{mm}$'s problem is negatively sloped, so $BD_{mm}$ would reduce its repo rate (and its inventories) to attain the maximum value of its problem at $r_{mm}^{A}(r_{rm})$. We call this the unconstrained maximum repo rate for $BD_{mm}$.

\subsection{BD\textsubscript{rm}'s problem}
\label{subsec:$BD_{rm}$ problem_certainty}

$BD_{rm}$'s problem, in its timing and mathematical structure, is similar to $BD_{mm}$'s problem. equation \ref{eq: $BD_{rm}$ problem} is $BD_{rm}$'s problem, contingent on its forecast of $r_{mm}$.

\begin{equation}
\label{eq: $BD_{rm}$ problem}
\argmax_{r_{rm}}\frac{1}{2}(r_{rm} - r_{mm})^{+}Q - gap(BD_{rm})
\end{equation}

Subject to $Q = \text{min}\{D_{r_{mm}}, S_{r_{rm}}\}$. $r_{rm} \in [r_{mm}, r_{b}]$ is the domain of $r_{rm}$\footnote{$r_{b}$ is the upper bound of the domain of $r_{rm}$ because it is the bank lending rate at which the $RM$'s can borrow. The domain follows from the requirement that $BD_{rm}$'s profit is non-negative.}.

In its client market, $BD_{rm}$ conveys $M$ to $RM$ as  a repo loan secured by collateral of $T$. $BD_{rm}$ carries the $T$ into the \idt. The first argument of $BD_{rm}$'s problem is its profit from the interdealer repo trade. The second term, $gap(BD_{rm}))$, represents $BD_{rm}$'s excess inventories; is its loss incurred on borrowings of $M$ from a bank used to extend the loan to $RM$ in its client market, that it is unable to refinance in the \idt. 

To analyze $BD_{rm}$'s problem we partition it into two parts; (i) the case where $BD_{rm}$ carries excess inventories and (ii) the case where $BD_{rm}$ does not carry excess inventories. We appeal to the demonstration for $BD_{mm}$ in Section \ref{subsec: $BD_{mm}$ carries excess inventories of $M$ into the interdealer trade. $gap(BD_{mm}) >0$} as sufficient to show that $BD_{rm}$ carries no excess inventory in equilibrium. Proposition \ref{corr:no excess inventory in equilibrium} follows from the fact that neither $BD$ will carry excess inventory into the \idt. 

\subsubsection*{BD\textsubscript{rm} carries no excess inventories of \textit{T} into the interdealer trade. gap\textsubscript{S,Q} = 0}

The domain of the problem is $r_{rm}\in [r_{mm}, r_{b}]$ such that $S_{r_{rm}}\leq D_{r_{mm}}$. It is a compact interval\footnote{The proof of compactness is nearly identical to the proof for the case of $BD_{mm}$ above.}. $BD_{rm}$'s problem is 

\begin{equation}
\argmax_{r_{rm}}\frac{1}{2}(r_{rm} - r_{mm})S_{r_{rm}}
\end{equation}

$BD_{rm}$'s problem is a continuous concave function. The concavity reflects the tension between the value of increasing transaction amount $S(r_{rm})$ against the decreased return, $r_{rm}$, required to obtain it, which reduces profit margin $\frac{1}{2}(r_{rm} - r_{mm})$.  

\subsubsection*{Characterizing BD\textsubscript{rm}'s decision problem}

\begin{figure}[H]
\begin{center}
\includegraphics[page=1,width=0.55\textwidth,height = 0.25\textheight]{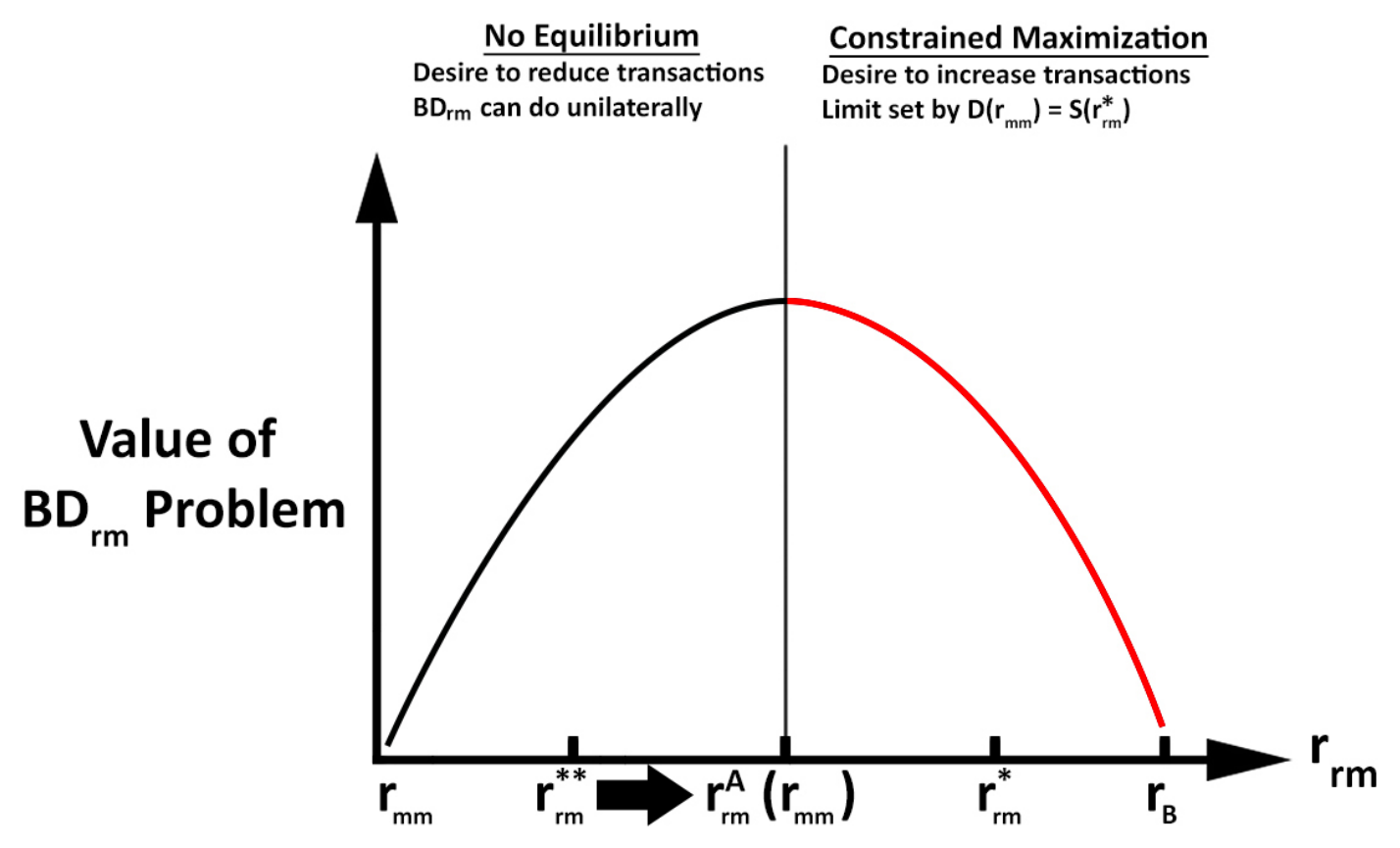}
\caption{$BD_{rm}$ Problem}
\label{fig:$BD_{rm}$ Problem}
\end{center}
\end{figure}
%revise figure

Figure \ref{fig:$BD_{rm}$ Problem} shows the curve defining the value of $BD_{rm}$'s problem at different repo rates for a given value of $r_{mm}$.

The value if $BD_{rm}$'s problem is zero at $r_{rm} = r_{mm}$ and at $r_{rm} = r_{b}$. The concavity of the problem ensures that it is single-peaked between these two endpoints. This peak is depicted in Figure \ref{fig:$BD_{rm}$ Problem} as $r_{rm}^{A}(r_{mm})$. $BD_{rm}$ will choose the maximum repo rate that (i) is not below $r_{rm}^{A}(r_{mm})$ and (ii) does not cause $gap(BD_{rm})) >0$. To see (i), note that $BD_{rm}$'s problem is positively sloped to the left of $r_{rm}^{A}(r_{mm})$ and it reduces inventory by increasing $r_{rm}$. Therefore, it it will increase profit by increasing $r_{rm}$, since $gap(BD_{rm}))$ will shrink. (ii) reflects that $BD_{mm}$ will not carry excess inventory in equilibrium.

To gain further insight we evaluate $BD_{rm}$'s decision at repo rates on either side of the maximum. Suppose $r_{rm}^{*}$ in Figure \ref{fig:$BD_{rm}$ Problem} is the repo rate at which $gap(BD_{mm}) = 0$. $BD_{rm}$ will not reduce $r_{rm}$, as doing so would cause $gap(BD_{rm}))$ to increase above zero. Since the value of $BD_{rm}$'s problem is decreasing at $r_{mm}^{*}$, $BD_{rm}$ maximizes its profit by choosing $r_{rm}^{*}$, which we call a constrained outcome repo rate for $BD_{rm}$\footnote{For $BD_{rm}$, any repo rate weakly above $r_{rm}^{A}(r_{mm})$ is a possible constrained equilibrium repo rate.}. On the other hand, at $r_{rm}^{**}$ in Figure \ref{fig:$BD_{rm}$ Problem} $BD_{rm}$'s problem is positively sloped, so $BD_{rm}$ would increase its repo rate to $r_{rm}^{A}(r_{mm})$. We call this the unconstrained maximum repo rate for $BD_{rm}$.

\subsection{No excess inventories}

The foregoing intuitive discussion can be formalized into a Proposition.

\begin{proposition*}[No excess inventory in \idt Equilibrium]
\label{prop:no excess inventory in equilibrium}
$D_{r_{mm}} = S_{r_{rm}}$ in equilibrium, i.e. neither $BD$ will carry excess inventories into the interdealer trade.
\end{proposition*}

\begin{proof}
See Appendix \ref{app:excess_inventories}
\end{proof}

\subsection{Application to a final-sale market}

The model presented in this section and discussed below can be applied to a final-sale market, such as the secondary market for US Treasuries. The translation is as follows. $MM$ becomes the ultimate buyer. $RM$ becomes the ultimate seller. $r_{mm}$ becomes the unit price $p_{mm}$ at which $MM$ purchases $T$ from $BD_{mm}$. $r_{rm}$ becomes the unit price $p_{rm}$ at which $RM$ sells $T$ to $BD_{rm}$. $r_{b}$ becomes $p_{bd}$ becomes the unit price at which $BD_{mm}$ purchases $T$ from $BD_{rm}$. It is left as an exercise for the reader to confirm that the analysis above applies to a final-sale market.

% === END: A_Model_of_Intermediated_Repo_Trade-2.tex ===

% === INLINED: Multiple_Equilibrium-2.tex ===
\section{Multiple Equilibrium}
\label{sec:Multiple Equilibirum}

In this section we prove there are multiple pure strategies equilibrium client repo rates chosen by \bd{s}. We first state definitions and concepts that are used in the proof of existence and multiplicity of  equilibrium and then we state the proofs. 

\subsection{Unique balance and the definition of equilibrium}

Restrictions are placed on the set of feasible equilibrium client repo rate pairs $BD$ by the non-negative profit requirement and  Proposition \ref{corr:no excess inventory in equilibrium}, which states that neither $BD$ will carry excess inventories into the \idt in equilibrium. Possible equilibria are restricted to repo rate pairs $\{r_{mm}, r_{rm}\}$ such that (i) $r_{mm} \leq r_{rm}$ and (ii) $D_{r_{mm}} = S_{r_{rm}}$. We call the satisfaction of (i) and (ii) "balanced trade". Figure \ref{fig:repo rates_balanced trade surface} displays the balanced trade surface (or manifold). The surface is downward sloping due to the fact that $D' > 0$ and $S' <0$. The intuition from Figure \ref{fig:repo rates_balanced trade surface} is that, when one repo rate is moved in one direction, up (down), the other repo rate must be moved in the opposite direction, down(up), to maintain balanced trade. 

In Figure \ref{fig:repo rates_balanced trade surface}, $\hat{r}$ denotes the identical repo rate for $BD_{mm}$ and $BD_{rm}$ at which the amount of inter-dealer trade is maximized\footnote{The proof that trade is maximized at a common repo rate is as follows. We know from (i) that $r_{mm}\leq r_{rm}$. Suppose the inequality is strict. Let $r_{rm} - r_{mm} = 2\epsilon$. Since $D_{r_{mm}}$ and $S_{r_{rm}}$ are continuous functions respectively increasing and decreasing in their repo rates, there is some $\delta > 0$ such that $D_{r_{mm} + \epsilon} - D_{r_{mm}} \leq \delta$ and $S_{r_{rm} - \epsilon} - S_{r_{rm}} \leq \delta$. Therefore, there must be some $\gamma < \epsilon$ such that $D_{r_{mm} + \gamma} = S_{r_{rm} - \epsilon}$. The volume of trade is increased and balanced and  $r_{mm}+ \gamma\leq r_{rm} - \epsilon$.}. At any point on the curve that is below and to the right of $\hat{r}$, $r_{mm} > r_{rm}$. Both $BD$'s would earn negative profit at any such repo rate pair. Therefore, equilibrium repo rate pairs occur only on the segment of the curve in Figure \ref{fig:repo rates_balanced trade surface} lying to the left and above $\hat{r}$.

\begin{definition}[Unique Balance]
\label{def:ub}
For each $r_{mm}(r_{rm})$ there is a unique $r_{rm}(r_{mm})$ for which $D_{r_{mm}} =  S_{r_{rm}}$.  We call this property unique balance.
\end{definition}

\begin{figure}[H]
\begin{center}
\includegraphics[page=1,width=0.45\textwidth,height = 0.25\textheight]{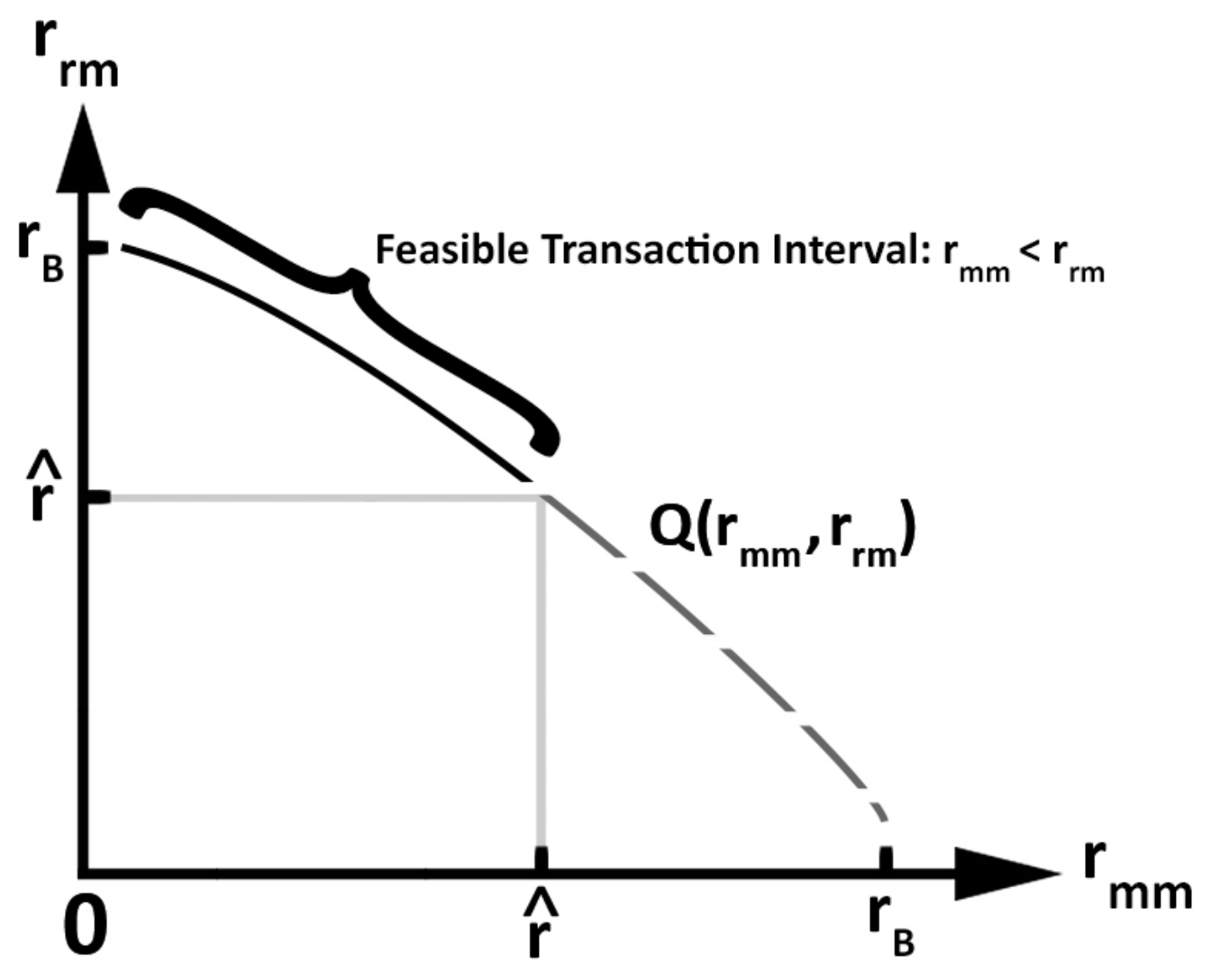}
\end{center}
\caption{Balanced Inter-dealer Trade Surface}
\label{fig:repo rates_balanced trade surface}
\end{figure}
%relabel figure

Definition \ref{def:eqm_pk} states the equilibrium concept that we use.

\begin{definition}[\idt Equilibrium under perfect knowledge]
\label{def:eqm_pk}
A repo rate pair $\{r_{mm}, r_{rm}\}$ is an equilibrium if the following two conditions hold;\\
(i) $r_{mm}$ solves $BD_{mm}$'s problem, given its knowledge of the realizations of $\{\theta_{b},\theta_{s} \}$ and $BD_{rm}$'s choice of $r_{rm}$.\\
(ii) $r_{rm}$ solves $BD_{rm}$'s problem, given its knowledge of the realization of $\{\theta_{b}, \theta_{s} \}$ and $BD_{mm}$'s choice of $r_{mm}$. 
\end{definition}

\noindent\textbf{Peak maps.} For each $r_{rm}$, let $r^{A}_{mm}(r_{rm})$ denote the unique maximizer of $BD_{mm}$'s no-inventory objective given $r_{rm}$ (the peak labeled $r_A(r_{rm})$ in Figure \ref{fig:$BD_{mm}$ problem}). For each $r_{mm}$, let $r^{A}_{rm}(r_{mm})$ denote the unique maximizer of $BD_{rm}$'s no-inventory objective given $r_{mm}$ (the peak labeled $r_A(r_{mm})$ in Figure \ref{fig:$BD_{rm}$ Problem}).

\begin{proposition**}[Nash equilibrium characterization]
\label{prop:NE_characterization}
Consider a repo-rate pair $(r_{mm},r_{rm})$ satisfying balanced trade:
\[r_{mm} \le r_{rm}\qquad\text{and}\qquad D(r_{mm}) = S(r_{rm}).\]
Then $(r_{mm},r_{rm})$ is a (pure-strategy) Nash equilibrium in client repo rates (Definition \ref{def:eqm_pk}) if and only if
\[r_{mm} \le r^{A}_{mm}(r_{rm})\qquad\text{and}\qquad r_{rm} \ge r^{A}_{rm}(r_{mm}).\]
Equivalently, along the balanced-trade manifold, each dealer must lie weakly on the inventory-reducing side of its no-inventory optimum.
\end{proposition**}

\begin{proof}
Fix a balanced pair $(r_{mm},r_{rm})$ with $D(r_{mm})=S(r_{rm})$. By Proposition 1 and the inventory-penalty structure, any unilateral deviation that increases a dealer's client-side quantity creates unsold inventory (since the counterparty's rate is fixed) and is not profitable. Hence it suffices to check deviations that reduce quantity:

$BD_{mm}$ lowering $r_{mm}$ and $BD_{rm}$ raising $r_{rm}$.

In the no-inventory region, $BD_{mm}$'s payoff as a function of $r_{mm}$ (holding $r_{rm}$ fixed) is strictly concave with unique maximizer $r^{A}_{mm}(r_{rm})$. Therefore $BD_{mm}$ has no profitable quantity-reducing deviation if and only if $r_{mm}$ lies weakly to the left of the peak, i.e.\ $r_{mm}\le r^{A}_{mm}(r_{rm})$. Similarly, $BD_{rm}$'s no-inventory payoff as a function of $r_{rm}$ (holding $r_{mm}$ fixed) is strictly concave with unique maximizer $r^{A}_{rm}(r_{mm})$, so $BD_{rm}$ has no profitable quantity-reducing deviation if and only if $r_{rm}$ lies weakly to the right of the peak, i.e.\ $r_{rm}\ge r^{A}_{rm}(r_{mm})$. Combining these yields the stated equivalence.
\end{proof}

\subsection{Concepts antecedent to the proof of multiple equilibrium}

In the section we present definitions and propositions that are used in the proof of multiple equilibrium.

\subsubsection*{Constrained equilibrium}

\begin{definition}[Constrained equilibrium]\label{def:constrained_equilibrium}
A balanced Nash equilibrium $(r_{mm},r_{rm})$ is called unconstrained if both dealers attain their no-inventory optima,
\[r_{mm} = r^{A}_{mm}(r_{rm})\qquad\text{and}\qquad r_{rm} = r^{A}_{rm}(r_{mm}).\]
Otherwise it is called constrained. Equivalently, by Proposition\ref{prop:NE_characterization}, a constrained equilibrium is a balanced pair satisfying
\[r_{mm} \le r^{A}_{mm}(r_{rm}),\quad r_{rm} \ge r^{A}_{rm}(r_{mm}),\quad\text{with at least one inequality strict.}\]
\end{definition}

A $BD$ is constrained if, in equilibrium, the derivative of its problem is not zero. An outcome where each $BD$ attains the maximum of its unconstrained problem is an "unconstrained equilibrium". It is immediate that zero trade is an equilibrium. 

\subsubsection*{Notional complementarity}

\begin{proposition***}[Notional complementarity]
\label{prop: Notional complementarity}
In the $BD$ problems, an increase in $r_{rm}$ increases the unconstrained maximum repo rate of $BD_{mm}$'s problem (i.e. shifts $r_{mm}^{A}$ in Figure \ref{fig:$BD_{mm}$ problem} to the right). A decrease in $r_{mm}$ reduces the unconstrained maximum repo rate of $BD_{rm}$'s problem (i.e. shifts $r_{rm}^{A}$ in Figure \ref{fig:$BD_{rm}$ Problem} to the left).
\end{proposition***}

\begin{proof}
For BD$_{mm}$ in the no-inventory region, the objective is
$\phi(r_{mm};r_{rm})=\tfrac12(r_{rm}-r_{mm})D(r_{mm})$, strictly concave in $r_{mm}$, so the maximizer
$r_{mm}^A(r_{rm})$ is unique and interior. The FOC is
\[
F(r_{mm},r_{rm}) := (r_{rm}-r_{mm})D'(r_{mm})-D(r_{mm})=0
\quad\text{at } r_{mm}=r_{mm}^A(r_{rm}).
\]
By the Implicit Function Theorem,
\[
\frac{d r_{mm}^A}{d r_{rm}}
= -\frac{\partial F/\partial r_{rm}}{\partial F/\partial r_{mm}}
= -\frac{D'(r_{mm}^A)}{(r_{rm}-r_{mm}^A)D''(r_{mm}^A)-2D'(r_{mm}^A)}.
\]
Since $D'>0$ and $D''<0$, the numerator is positive and the denominator is negative, hence
$\frac{d r_{mm}^A}{d r_{rm}}>0$.

For BD$_{rm}$ in the no-inventory region, the objective is
$\psi(r_{rm};r_{mm})=\tfrac12(r_{rm}-r_{mm})S(r_{rm})$, strictly concave in $r_{rm}$, so the maximizer
$r_{rm}^A(r_{mm})$ is unique and interior. The FOC is
\[
G(r_{rm},r_{mm}):=S(r_{rm})+(r_{rm}-r_{mm})S'(r_{rm})=0
\quad\text{at } r_{rm}=r_{rm}^A(r_{mm}).
\]
IFT gives
\[
\frac{d r_{rm}^A}{d r_{mm}}
= -\frac{\partial G/\partial r_{mm}}{\partial G/\partial r_{rm}}
= -\frac{-S'(r_{rm}^A)}{2S'(r_{rm}^A)+(r_{rm}^A-r_{mm})S''(r_{rm}^A)}.
\]
Because $S'<0$ and $S''<0$, the numerator is positive and the denominator is negative, so
$\frac{d r_{rm}^A}{d r_{mm}}>0$.
\end{proof}

We call the result of this proposition "notional" complementarity because the solution to the unconstrained $BD$ problems react like strategic compliments. An increase in $BD_{rm}$'s repo rate increases the repo rate spread and induces $BD_{mm}$ to desire to increase its volume of trade by increasing its repo rate. Similarly, a reduction in $BD_{mm}$'s repo rate increases the repo rate spread and induces $BD_{mm}$ to desire to increase its volume of trade by reducing its repo rate. However, the friction imposed by the zero inventory constraint prevents these adjustments from taking place. We state a definition derived from notional complementary that is used in the proof of equilibrium.

\begin{definition}{concavity/complementarity}
From the strict concavity of the $BD$ problems and the positive slope of the implicit functions in notional complementarity, for each $r_{mm}(r_{rm})$ there is a unique  $r_{rm}^{A}(r_{mm}^{A})$, i.e. unique peaks to the $BD$ maximization problems. We call this property ''concavity/ complementarity''.
\end{definition}

\subsection{Proof of existence and multiplicity of equilibrium}

We make use of the following relationships which are derived from the analysis of the $BD$ problems in Section \ref{sec:A Model of Intermediated Repo Trade}.\footnote{Although we have shown that zero trade is an equilibrium, we provide another proof of existence of equilibrium and use the machinery employed in that proof to prove multiplicity. }

(i) When $r_{mm}$ is decreased to reduce inventory of $M$, $BD_{rm}$'s unconstrained optimal repo rate is pushed down, but it must do the opposite and increase $r_{rm}$ in order to reduce its inventory of $T$ to match the reduced volume of $M$. Hence, $r_{rm}$ will increase rightward distance (or reduce leftward distance)  relative to the peak, $r_{rm}^{A}$, of $BD_{rm}$'s problem.

(ii)  When $r_{rm}$ is increased to reduce inventory of $T$, $BD_{mm}$'s unconstrained optimal repo rate is pushed down, but it must do the opposite, and increase $r_{mm}$ in order to reduce its inventory of $M$ to match the reduced volume of $T$. Hence, $r_{mm}$ will increase leftward distance (or reduce rightward distance) from the peak, $r_{mm}^{A}$, of $BD_{mm}$'s problem.

\begin{lemma}[Unique balance existence]\label{lem:unique-balance}
Assume $D:[0,\hat r]\to\mathbb{R}_+$ is continuous and strictly increasing, and
$S:[\hat r,r_b]\to\mathbb{R}_+$ is continuous and strictly decreasing. Assume further that
\[D(0)=0,\qquad S(r_b)=0,\qquad \text{and}\qquad D(\hat r)=S(\hat r)=:T_{\max}.\]
Then for every $r_{mm}\in[0,\hat r]$ there exists a unique $r_{rm}\in[\hat r,r_b]$ such that
$D(r_{mm})=S(r_{rm}).$ Equivalently, the map $B:[0,\hat r]\to[\hat r,r_b]$ defined by
$B(r_{mm}) := S^{-1}\!\bigl(D(r_{mm})\bigr)$
is well-defined (and unique).
\end{lemma}

\begin{proof}
Fix $r_{mm}\in[0,\hat r]$ and set $q:=D(r_{mm})\in[0,T_{\max}]$ (since $D$ is continuous, strictly increasing, and $D(0)=0$, $D(\hat r)=T_{\max}$). Because $S:[\hat r,r_b]\to[T_{\max},0]$ is continuous and strictly decreasing, it is invertible on $[0,T_{\max}]$ with inverse $S^{-1}:[0,T_{\max}]\to[\hat r,r_b]$.
Define $r_{rm}:=S^{-1}(q)$. Then $S(r_{rm})=q=D(r_{mm})$. Uniqueness follows from strict monotonicity of $S$.
\end{proof}

% \begin{proof}
% Fix any $r_{mm}\in[0,\hat r]$ and set $q:=D(r_{mm})$. Since $D$ is strictly increasing and
% $D(0)=0$, $D(\hat r)=T_{\max}$, we have $q\in[0,T_{\max}]$.
% Define the continuous function $g:[\hat r,r_b]\to\mathbb{R}$ by
% \[g(r):=S(r)-q.\]
% By the boundary conditions on $S$ and the fact that $q\in[0,T_{\max}]$,
% \[g(\hat r)=S(\hat r)-q=T_{\max}-q\ge 0,\qquad
% g(r_b)=S(r_b)-q=0-q\le 0.\]
% By the Intermediate Value Theorem, there exists at least one $r_{rm}\in[\hat r,r_b]$ such that
% $g(r_{rm})=0$, i.e.\ $S(r_{rm})=q=D(r_{mm})$. To prove uniqueness, suppose $r_{rm}^1$ and $r_{rm}^2$ in $[\hat r,r_b]$ both satisfy
% $S(r_{rm}^1)=S(r_{rm}^2)=q$. Since $S$ is strictly decreasing on $[\hat r,r_b]$, it is injective there, hence $r_{rm}^1=r_{rm}^2$. Therefore the balancing $r_{rm}$ is unique. Finally, since $S$ is continuous and strictly monotone on $[\hat r,r_b]$, its inverse $S^{-1}$ exists onm$[0,T_{\max}]$, so the balanced rate can be written uniquely as
% $B(r_{mm})=S^{-1}(D(r_{mm}))$.
% \end{proof}

\begin{theorem*}[Existence of Inter-Dealer Pure Strategy Equilibrium Client repo rate Pairs]
\label{thm:Existence of equilibrium repo rate pairs}
There is a mapping $r_{mm} \mapsto \{r_{mm}, r_{rm}\}$ such that any $r_{mm}\in [0,\hat{r}]$ maps to a repo rate pair $\{r_{mm}', r_{rm}'\}$ that is on the Feasible Transaction Interval of the Balanced Trade Surface in Figure~\ref{fig:repo rates_balanced trade surface} and is an equilibrium.\footnote{An equivalent theorem proves the mapping $r_{rm} \mapsto \{r_{mm}, r_{rm}\}$.} 
\end{theorem*}

\begin{proof}
We use an algorithm to find an equilibrium.

\noindent\textbf{Step 1.} Fix an arbitrary client money-market repo rate $r_{mm}\in[0,\hat{r}]$. Interpret this as the initial quote $r_{mm}$ of dealer $BD_{r_{mm}}$. We will construct a pair $\{r_{mm}', r_{rm}'\}$ and thereby a mapping $r_{mm}\mapsto\{r_{mm}', r_{rm}'\}$ with the desired properties.

\medskip

\noindent\textbf{Step 2.} Choose a repo-market client rate $r_{rm}$ such that money-market demand and repo-market supply balance,
\[D_{r_{mm}} = S_{r_{rm}}.\]
By the unique balance property, this condition defines a one-to-one mapping between $r_{mm}$ and $r_{rm}$. At this point the pair $(r_{mm}, r_{rm})$ lies on the Balanced Trade Surface, as depicted in Figure~\ref{fig:repo rates_balanced trade surface}. Steps 2–6 will place $r_{mm}$ weakly to the left of its unconstrained maximum (see Figure~\ref{fig:seq of $r_{mm}$ steps}).

\medskip

\noindent\textbf{Step 3.} If $r_{mm}$ is above the repo rate that maximizes $BD_{1i}$’s objective, we proceed to adjust $r_{mm}$ as in Step~4. If instead $r_{mm}$ is already on the constrained side of its maximizer (i.e., weakly to the left of the peak in Figure~\ref{fig:seq of $r_{mm}$ steps}), then no adjustment to $r_{mm}$ is needed and we may skip directly to Step~6.

\medskip

\noindent\textbf{Step 4.} When $r_{mm}$ is above its unconstrained maximizer, choose $r_{mm}$ again so that
\[r_{mm} = r_{mm}^{A}(r_{rm}),\]
where $r_{mm}^A(r_{rm})$ denotes the argmax of $BD_{r_{mm}}$’s problem given $r_{rm}$. By concavity and complementarity, this map is one-to-one. Economically, this moves $BD_{r_{mm}}$’s client repo rate leftward to the peak of its payoff function. Since the money-market demand function $D_{r_{mm}}$ is increasing in $r_{mm}$, this leftward move reduces \m, the inventory carried by $BD_{r_{mm}}$.

\medskip

\noindent\textbf{Step 5.} Given the new $r_{mm}$ from Step~4, choose $r_{rm}$ again so that
\[D_{r_{mm}} = S_{r_{rm}}.\]
By the unique balance property, this mapping from $r_{mm}$ to $r_{rm}$ is again one-to-one. Because Step~4 reduced $BD_{r_{mm}}$’s inventory, we must reduce $BD_rm$’s inventory correspondingly. This requires increasing $r_{RM_j}$, since the repo-market supply function $S_{r_{rm}}$ is decreasing in $r_{rm}$. By the notional complementarity property, this increase in $r_{rm}$ raises $r_{mm}^{A}$ while $r_{mm}$ itself remains fixed. Hence $r_{mm}$ remains in the constrained outcome region of $BD_{r_{mm}}$’s problem, on the left-hand side of the peak shown in Figure~\ref{fig:seq of $r_{mm}$ steps}, while the balance condition keeps the pair on the Balanced Trade Surface.

\medskip

\noindent\textbf{Step 6.} If $r_{mm}$ was already in a constrained outcome position at Step 2, then Steps~3–5 are not applied, and $r_{mm}$ is left unchanged. In either case, after Step~6 we have a pair $(r_{mm}, r_{rm})$ on the Balanced Trade Surface with $BD_{r_{mm}}$ in its constrained region (weakly to the left of its unconstrained maximizer).

\medskip

\noindent\textbf{Step 7.} We now adjust $r_{rm}$ while keeping $r_{mm}$ in the constrained region. Steps~7–10 place $r_{rm}$ weakly to the right of its unconstrained maximum (see Figure~\ref{fig:seq of $r_{rm}$ steps}). If $r_{rm}$ is below the repo rate that maximizes $BD_{rm}$’s problem, we proceed to Step~8. If instead $r_{rm}$ is already in the constrained maximization region (weakly to the right of its peak), we skip directly to Step~10.

\medskip

\noindent\textbf{Step 8.} When $r_{rm}$ is below its unconstrained maximizer, set
\[r_{rm} = r_{rm}^{A}(r_{mm}),\]
where $r_{rm}^{A}(r_{mm})$ is the argmax of $BD_{rm}$’s objective for given $r_{mm}$. By concavity and complementarity, this map is one-to-one. This step moves $BD_{rm}$’s client repo rate rightward to the peak of its problem. Since the supply function $S(\cdot)$ is decreasing in $r_{rm}$, raising $r_{rm}$ reduces $T$, the inventory carried by $BD_{rm}$. By the notional complementarity property, this increase in $r_{rm}$ raises $r_{mm}^{A}$ while $r_{mm}$ remains fixed. Thus $r_{mm}$ is on the left-hand side of its peak with an excess inventory of $M$, while $r_{rm}$ is at the unconstrained peak of $BD_{rm}$’s problem.

\medskip

\noindent\textbf{Step 9.} Given the updated $r_{rm}$ from Step~8, choose $r_{mm}$ such that
\[D_{r_{mm}} = S_{r_{rm}}.\]
This balance condition again defines a one-to-one mapping by uniqueness. To offset the inventory reduction of $BD_{rm}$ induced in Step~8, we decrease $r_{mm}$ so as to reduce $BD_{r_{mm}}$’s inventory. By the notional complementarity property, this decrease in $r_{mm}$ shifts $r_{rm}^{A}$ leftward in $BD_{rm}$’s problem while leaving $r_{rm}$ fixed. Consequently, $r_{rm}$ remains in the constrained maximization region of $BD_{rm}$ (on the right-hand side of the peak in Figure~\ref{fig:seq of $r_{rm}$ steps}), and $r_{mm}$—already on the left-hand side of its peak—moves further away from its unconstrained maximizer. At this point, both $BD_{r_{mm}}$ and $BD_{rm}$ are in their constrained optimization regions, and neither carries excess inventory. The pair $\{r_{mm}, r_{rm}\}$ is therefore an equilibrium point on the Balanced Trade Surface, and by the inventory constraints it lies on the Feasible Transaction Interval.

\medskip

\noindent\textbf{Step 10.} If, in Step~7, $r_{rm}$ was already in a constrained maximization position, then Steps~8–9 are not needed. The existing pair $\{r_{mm}, r_{rm}\}$ already has both dealers in their constrained regions with balanced inventories, so it is an equilibrium, again lying on the Feasible Transaction Interval of the Balanced Trade Surface.

\medskip

Finally, denote by $r_{mm}'$ and $r_{rm}'$ the equilibrium client repo rates obtained at the end of this procedure for a given initial $r_{mm}$. The construction above defines a mapping $r_{mm} \mapsto \{r_{mm}', r_{rm}'\}$ from $[0,\hat{r}]$ into equilibrium repo rate pairs on the Feasible Transaction Interval of the Balanced Trade Surface, which proves the theorem.
\end{proof}

\begin{figure}[H]
\begin{center}
\includegraphics[page=1,width=0.55\textwidth,height = 0.25
\textheight]{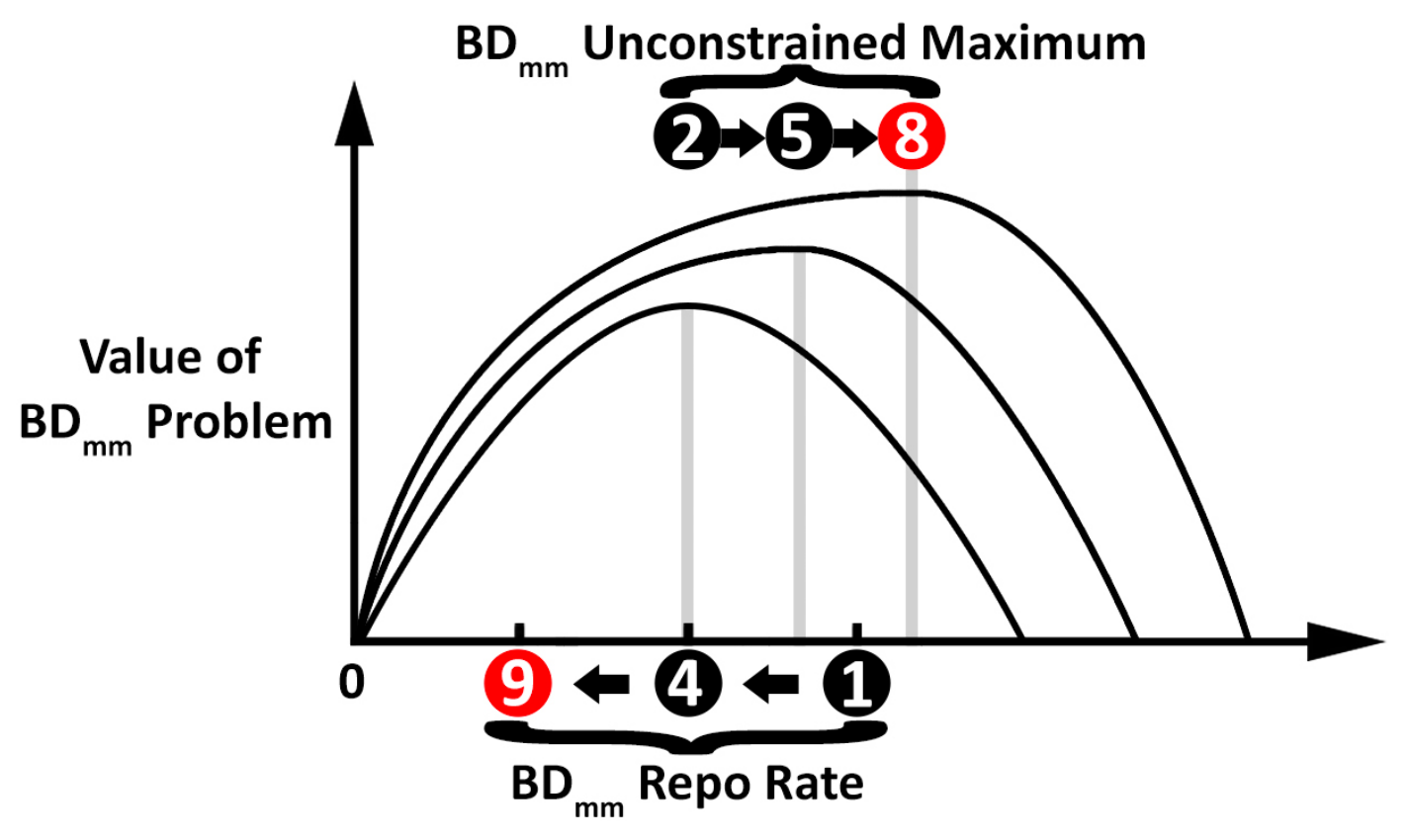}
\end{center}
\caption{Sequence of $r_{mm}$ steps}
\label{fig:seq of $r_{mm}$ steps}
\end{figure}

Figure \ref{fig:seq of $r_{mm}$ steps} depicts the proof sequence of $BD_{mm}$ argmax (notional) repo rates (top). They are identified by gray vertical lines and circled numbers at the top. It also depicts the proof sequence of actual repo rates (bottom). They ae identified by marks in the horizontal axis and circled numbers at the bottom. In the sequence of proof steps $BD_{mm}$'s argmax repo rate progressively increases and its actual repo rate progressively decreases. $BD_{mm}$ is moved to its argmax repo rate in step 4 and stays on the LHS of its argmax repo rate in all subsequent proof steps.               

\begin{figure}[H]
\begin{center}
{\includegraphics[page=1,width=0.55\textwidth,height = 0.25
\textheight]{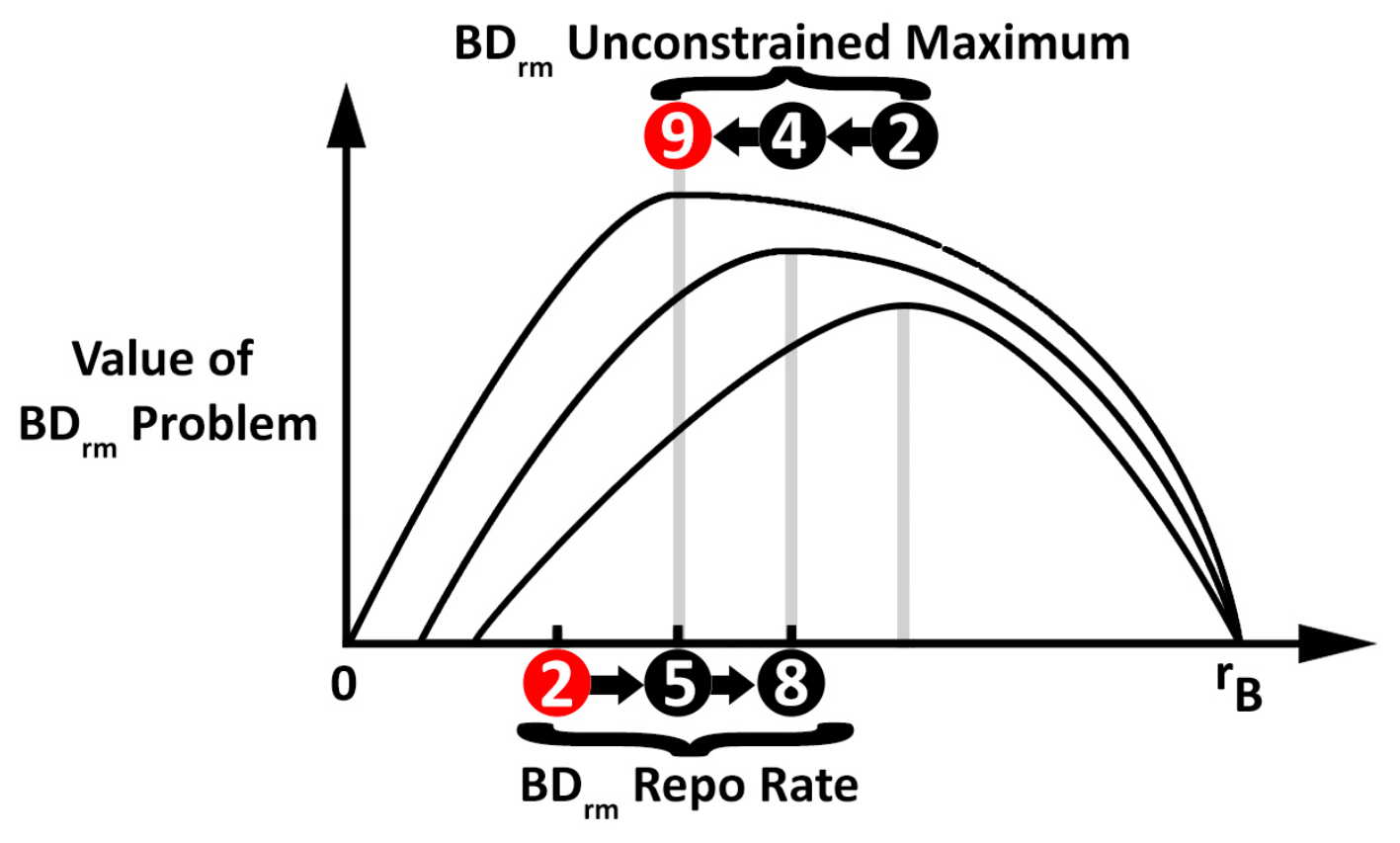}}
\end{center}
\caption{Sequence of $r_{rm}$ steps}
\label{fig:seq of $r_{rm}$ steps}
\end{figure}

Figure \ref{fig:seq of $r_{rm}$ steps} depicts the proof sequence of $BD_{rm}$ argmax (notional) repo rates (top). They are identified by gray vertical lines and circled numbers at the top. It also depicts the proof sequence of actual repo rates (bottom). They ae identified by marks in the horizontal axis and circled numbers at the bottom. In the sequence of proof steps $BD_{rm}$'s argmax repo rate progressively decreases and its actual repo rate progressively increases. $BD_{rm}$ is moved to its argmax repo rate in step 8 and stays on the RHS of its argmax repo rate in all subsequent proof steps.

Our next theorem proves there is a continuum of equilibrium. 

\begin{theorem**}[Multiplicity of Inter-Dealer Equilibrium Repo-Rate Pairs]\label{thm:mult}
There are multiple equilibrium repo-rate pairs $\{r_{mm},r_{rm}\}$.
In fact, there exists $\varepsilon>0$ such that for every $r_{mm}\in(0,\varepsilon]$ the balanced pair
\[
\bigl(r_{mm},\,B(r_{mm})\bigr)i
\quad\text{with}\quad
D(r_{mm})=S\!\bigl(B(r_{mm})\bigr)
\]
is a (pure-strategy) equilibrium. Hence the set of equilibria is a continuum.
\end{theorem**}

\begin{proof}
Let $B(r_{mm})$ be the unique-balance map defined by $D(r_{mm})=S(B(r_{mm}))$ (Lemma~1). Consider the
balanced zero-trade point $(0,B(0))=(0,r_b)$. At $r_{rm}=r_b$, BD$_{mm}$'s no-inventory objective $\phi(r_{mm};r_b)=\tfrac12(r_b-r_{mm})D(r_{mm})$ is strictly concave on $[0,r_b]$, satisfies
$\phi(0;r_b)=\phi(r_b;r_b)=0$, and is positive at some interior point, hence its unique maximizer $r_{mm}^A(r_b)$ lies in $(0,r_b)$. Thus $0<r_{mm}^A(B(0))$.
Similarly, at $r_{mm}=0$, BD$_{rm}$'s no-inventory objective
$\psi(r_{rm};0)=\tfrac12 r_{rm} S(r_{rm})$ is strictly concave on $[0,r_b]$ with $\psi(0;0)=\psi(r_b;0)=0$ and positive somewhere, so its unique maximizer $r_{rm}^A(0)$ lies in $(0,r_b)$, hence $B(0)=r_b>r_{rm}^A(0)$.

By continuity of $B(\cdot)$ and of the peak maps $r_{mm}^A(\cdot), r_{rm}^A(\cdot)$, there exists $\varepsilon>0$
such that for all $r_{mm}\in(0,\varepsilon]$,
\[r_{mm} < r_{mm}^A(B(r_{mm}))\quad\text{and}\quad B(r_{mm}) > r_{rm}^A(r_{mm}).\]

For such $r_{mm}$ we have
\[r_{mm} < r^{A}_{mm}(B(r_{mm})) \quad\text{and}\quad B(r_{mm}) > r^{A}_{rm}(r_{mm}).\]
Since $(r_{mm},B(r_{mm}))$ is balanced by construction, Proposition \ref{prop:NE_characterization} implies that $(r_{mm},B(r_{mm}))$ is a (pure-strategy) Nash equilibrium. The strict inequalities imply it is constrained in the sense of Definition \ref{def:constrained_equilibrium}. Finally, if $r_{mm}^1\ne r_{mm}^2$, then $D(r_{mm}^1)\ne D(r_{mm}^2)$ since $D$ is strictly increasing, so $B(r_{mm}^1)\ne B(r_{mm}^2)$ since $S$ is strictly monotone. Thus the equilibria form a continuum.
\end{proof}

\subsubsection{Funding liquidity}

It is remarked by market participants that a feature of the relationship formed between a \bd and its client is that the \bd will guarantee to transact a certain volume. In Appendix \ref{app:funding_liquidity} we extend our model to the case where a \bd commits to trade a fixed volume with its client. Funding commitments induce a feasible trade surface of client repo rates that can result in excess inventories. We prove that the main results of the analysis in the text continue to hold for the case of funding liquidity. 

\begin{proposition}[Existence of equilibrium with client funding commitments]\label{prop:A1_fixed} There is at least one Nash equilibrium in client repo rates $(r_{mm},r_{rm})$ when one or both broker-dealers is subject to a funding commitment (Definition~\ref{def:funding_commitment}).
\end{proposition}

\begin{proof}
See Appendix \ref{app:funding_liquidity}    
\end{proof}

% === END: Multiple_Equilibrium-2.tex ===

% === INLINED: Joint_Profit_Max-2.tex ===

\section{Joint Profit Maximization}
\label{sec:Joint Profit Maximization}

In this section we demonstrate that \bd joint profit maximization is an equilibrium. This results from the surplus allocation rule whereby each \bd receives a fixed percentage of the surplus generated by the spread between the borrower and lender repo rates. Any deviation from the joint profit maximization point $\{r_{rm}, r_{mm}\}$ will reduce profit for both \bd{s}.

\begin{figure}[H]
\begin{center}
\includegraphics[page=1,width=0.5\textwidth,height = 0.3
\textheight]{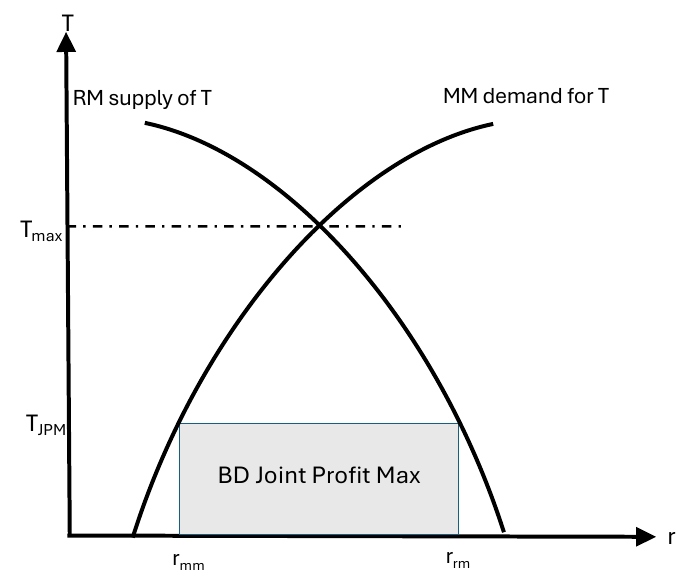}
\caption{$BD$ Joint Profit Maximization}
\label{fig:BD Joint Profit Maximization}
\end{center}
\end{figure}

From Figure \ref{fig:BD Joint Profit Maximization} it can be seen that joint $BD$ profit at any volume of trade $T$ is the area of the rectangle with height $T$ and sides $r_{mm} = D^{-1}(T)$ and $r_{rm} = S^{-1}(T)$. As $T$ increases, the height of the rectangle rises and the width falls. At some volume of trade the area will reach its maximum. To analyze $BD$ welfare, we conceive of a planner who  controls repo rates (and thereby the volume of trade), which is the same instrument the $BD$'s control. Since $BD$'s are risk neutral and their profit is their utility, $BD$ utility is transferable. Therefore, maximizing $BD$ welfare is equivalent to maximizing joint $BD$ profits, which is represented by the area of the rectangle in Figure \ref{fig:BD Joint Profit Maximization}. The planner solves the following problem

\begin{equation}
\label{eq:JPmaxpbm}
\argmax_{r_{mm}, r_{rm}, T}(r_{rm} - r_{mm})T \; \text{subject to}\; D(r_{mm},\theta_{b}) = T \; \text{and}\; S(r_{rm},\theta_{s}) = T
\end{equation}

Noting that $r_{mm} = D^{-1}(T)$ and $r_{rm} =  S^{-1}(T)$, we rewrite the problem in terms of $T$.

\begin{equation}
\argmax_{T}(S^{-1}(T) - D^{-1}(T))T  
\end{equation}

The FOC is

\begin{equation}
\label{eq:JPM_FOC}
f(T) = \underset{repo\; rate\; spread}{\underbrace{S^{-1}(T) - D^{-1}(T})} + T\underset{slope\; divergence}{\underbrace{\left\{\frac{d S^{-1}(T)}{d T} - \frac{d D^{-1}(T)}{d T}\right\}}} = 0
\end{equation}

$f(T)$ is continuous C2 function, since $RM$ and $D$ are C2 functions. $D^{-1}$ and $S^{-1}$ are, respectively, decreasing and increasing convex functions, since $D$ and $RM$ are, respectively, monotone increasing and monotone decreasing concave functions. These observations imply the following relationships.

$\frac{d S^{-1}(T)}{d T} < 0$ and $\frac{d D^{-1}(T)}{d T} > 0 \implies \frac{d S^{-1}(T)}{d T} - \frac{d D^{-1}(T)}{d T} < 0$

$D^{-1}(T_{max}) = S^{-1}(T_{max})$ (by definition).

$T = 0 \implies r_{mm} = 0$ and $r_{rm} = r_{B} \implies S^{-1}(0) - D^{-1}(0)= r_{B} > 0$

These facts enable us to state a proposition concerning $BD$ joint profit maximization.

\begin{proposition****}[$BD$ Joint Profit Maximization Repo Rates -  Existence]
\label{prop: $BD$ Joint Profit Maximization Repo Rates -  Existence}
There is a  repo rate pair which attains $BD$ joint profit maximization. 
\end{proposition****}

\begin{proof}
See Appendix \ref{app:Proof_of_JPM_Equilibrium}
\end{proof}

We can prove that joint profit maximization is an equilibrium.

\begin{corollary*}[$BD$ Joint Profit Maximization Equilibrium]
\label{corr: $BD$ Joint Profit Maximization Repo Rates-Eqm}
There is a joint profit maximizing equilibrium.
\end{corollary*}

\begin{proof}
See Appendix \ref{app:Proof_of_JPM_Equilibrium}
\end{proof}

\subsection{Minimum spread}
\label{subsec:minimum_spread}

Section \ref{subsec:Hurdlerates} cited empirical evidence for the use by banks of a minimum spread in allocating balance-sheet to repo intermediation.
Suppose each $BD$ faces a hurdle marginal profit rate $\kappa \geq \kappa(BD)$ that it must achieve to enter into an inter-dealer trade. WLOG let $BD_{rm}$ face the higher minimum spread. 

% This contracts the feasible transaction interval on the balanced trade surface of Figure \ref{fig:repo rates_balanced trade surface} away from $\hat{r}$ to $1/2(r_{rm} - r_{mm}) \geq \kappa(BD_{rm})$ and pushes down $T_{max}$ in Figure \ref{fig:BD Joint Profit Maximization}. This limits the range of possible equilibria, but it does not otherwise change the problem. There is at least one \jpm volume of $T$. It is immediate that the proofs of Proposition \ref{prop: $BD$ Joint Profit Maximization Repo Rates -  Existence} and  Corollary \ref{corr: $BD$ Joint Profit Maximization Repo Rates-Eqm} apply to the contracted interval of feasible trade when there is a minimum spread imposed. We state this result in Corollary \ref{corr:HurdleRateJPM}

\begin{corollary**}[Joint Profit Maximization Equilibrium with a minimum spread]
\label{corr:HurdleRateJPM}
 When a minimum spread $1/2(r_{rm} - r_{mm}) \geq \kappa$ is imposed, there is a \jpm equilibrium.   
\end{corollary**}

\begin{proof}
Restrict the feasible set to balanced volumes $T$ such that $\tfrac12(S^{-1}(T)-D^{-1}(T))\ge \kappa$. Continuity of $S^{-1},D^{-1}$ implies the restricted objective is continuous on a compact set, so a maximizer exists (as in Proposition~5). The corresponding balanced rate pair is a constrained outcome and hence an
equilibrium (as in Corollary~1).
\end{proof}

We refer to this as the ``constrained \jpm'' equilibrium.
% === END: Joint_Profit_Max-2.tex ===

% === INLINED: Equilibrium_Under_Asymmetric_Information.tex ===
\section{Equilibrium Under Asymmetric Information}
\label{sec:Equilibrium Under Asymmetric Information}

Our model can be extended to the case where the client shock of one broker is unobserved by the other broker. Each $BD$ chooses its repo rate after observing its own client shock, but does not know the shock realization of its counterparty's client. In Appendix \ref{app:asymmetric_eqm} we prove there are multiple equilibria when this information pattern holds. We define an asymmetric equilibrium as follows.

\begin{definition}[Equilibrium under Asymmetric Information]

In a model where each $BD$ sets its repo rate after observing its own client shock, but not its counterparty's client shock, and the distribution of clients shocks is common knowledge, a pair of strategies $S_{1}, S_{2}$ is an equilibrium if the following holds;

$S_{1}:\theta_{mm} \implies r_{mm}$, and 

$S_{2}:\theta_{rm} \implies r_{rm}$

such that $S_{1}$ and $S_{2}$ are equilibrium strategies.
\end{definition}

Appendix \ref{app:asymmetric_eqm} contains proofs of the existence and multiplicity of equilibrium under asymmetric information. 

\begin{theorem***}[Existence of Asymmetric Equilibrium Strategies]
\label{thm: Existence of Asymmetric Equilibrium Strategies} 
There is an inventory $T$ such that the strategies of setting repo rates to achieve $T$ volume of transaction is an equilibrium. That is, there exists $T$ such that each dealer sets its client rate to induce client-side volume $T$ for every
realized type.

% $S_{1}:\theta_{mm} \implies r_{mm}$ such that $D_{r_{mm},\theta_{mm}} = T$ 

% $S_{2}:\theta_{rm} \implies r_{rm}$ such that $S_{r_{rm}, \theta_{rm}} = T$
\end{theorem***}

\begin{proof}
See Appendix \ref{app:existence_mult_assymetric}
\end{proof}

\begin{theorem****}[Multiplicity of Asymmetric Equilibrium Strategies]
There are multiple equilibrium strategies under asymmetric information involving inventory volumes of $T$, i.e. That is, there exist $T\neq T'$ such that both the $T$-targeting and the $T'$-targeting profiles are equilibria.% i.e. there are at least two values $T$ and $T^{'}$ such that (i) $T \neq T^{'}$ and (ii) there is an equilibrium in which each $BD$ transacts $T$ inventory with it client and there is another equilibrium in which each $BD$ transacts $T^{'}$ inventory with its client. 
\end{theorem****}

\begin{proof}
See Appendix \ref{app:existence_mult_assymetric}
\end{proof}

The proof is intricate, but the key insight is that $BD$'s can coordinate on inventories. There are multiple values of $T$ such that it is an equilibrium for $BD_{mm}$ to set $r_{mm}$ to convey $T$ to its client and for $BD_{rm}$ to set $r_{rm}$ to receive $T$ from its client, even though neither $BD$ knows the counterparty repo rates required to achieve that outcome.

\section{Balance-Sheet Costs and Trade Restrictions}
\label{sec:Balance_Sheet_Costs}

In this section we evaluate the effect of balance-sheet cost incurred by \bd{s} in restricting the volume of repo trade they can intermediate. The constraint arises from the interaction of accounting rules and bank capital regulations. A \bd engages in a matched-book trade where it passes on at each leg the money and collateral provided by the clients, retaining only its profit. A \bd does not need to hold an inventory of money or collateral between the \fl and \sel. Nevertheless, GAAP accounting rules require the \bd to increase its recorded assets by the full first-leg sale price of the \g (since repurchase agreements are treated as secured borrowings rather than sales under GAAP) (Bowman et.al. \cite{BowmanHuhInfante2024}). At the same time, bank capital regulations – notably the supplementary leverage ratio (SLR) – limit the volume of assets a bank can record on its \bs, as the leverage ratio serves as a non-risk-based constraint on overall leverage (\citep{BIS2014}). This creates a shadow price on the allocation of \bs assets – even where there is no underlying inventory.\footnote{See Aronoff et al. \cite{AronoffrepoBS2025} for a detailed explanation of repo accounting rules and capital regulations and an analysis of their interaction.} Denote $c$ the marginal shadow cost of \bs. This equates to $c = p_{1}$, where $p_{1}$ is the \fl unit price of \g. Let $p_{1} = p_{i}^{M} - h$ for all trades, where $p^{M}_{1}$ is the market price of \g at the \fl date and $h$ is the haircut.\footnote{This assumption reflects empirical U.S. Treasury repo market practice, where most such repo transactions are executed with very low or zero haircuts (Hempel et.al. \cite{Hempel2023}).} Each \bd makes one purchase of \g, so its shadow price is $c$ per unit of \g. We assume the bank allocates capacity to its \bd affiliate prior to the commencement of trading, in which case the allocation decision is based on the ex-ante expected profit of trading. Equation \ref{eq:capital-hurdle} is the hurdle rate of return for each $BD_{i}$, expressed as a constraint on the client repo rate spread.

\begin{equation}
\label{eq:capital-hurdle}
\mathbb{E}\!\left[\frac{1}{2}(r_{rm} - r_{mm})\right] \geq \kappa_{i}
\end{equation}

The feasible range of the spread is $[0, r_{B}]$ (Figure \ref{fig:repo rates_balanced trade surface}), which extends below the hurdle rate. An ex-ante distribution over multiple equilibrium repo rate pairs may yield an expectation that is below the hurdle rate, even though there are equilibria that exceed the hurdle. For example, suppose there are two equilibrium repo rate pairs $\{\{12,6\},\{8,6\}\}$, $c = 2.5$ and an equal ex-ante probability of each equilibrium. The ex-ante expected unit profit is 2, which is below the shadow price, even as one of the equilibria has a spread of 3. In this case no trade will occur even though there exist equilibria that satisfy both dealers’ hurdle constraints and would support positive intermediation. 

The no-trade outcome is therefore driven by equilibrium-selection risk under ex ante balance-sheet allocation, rather than by the absence of hurdle-feasible intermediation opportunities. From a regulatory perspective, this is the key coordination problem: the market can fail to intermediate not because all feasible trades are eliminated by regulation, but because dealers cannot commit ex ante to select a hurdle-feasible equilibrium from the set characterized in Theorem 2.

% === END: balance_sheet_costs.tex ===

% === INLINED: The_Smart_Contract_v4.tex ===

\section{A Smart Contract That Resolves Multiple Equilibrium}
\label{sec:SmartContract}

In this section, we describe a mechanism that resolves multiple equilibrium by committing the market to a hurdle-feasible equilibrium-selection rule. The mechanism coordinates all four participants to select a trade that satisfies both broker-dealers’ minimum spread requirements, thereby preventing the collapse to no trade that can arise when equilibrium-selection risk interacts with binding balance-sheet constraints. There are three additional attractive features of the \smc.

\begin{itemize}
\item There is no leakage of client schedules to the counterparty $BD$, or to any other party.

\item Truthful reporting to the \smc by each $BD$ of its client schedule and minimum spread is a dominant strategy for the former and a weakly dominant strategy for the latter. This minimizes the computational and knowledge burden imposed on the $BD$'s to implement their strategies. 

\item The \smc achieves these results without requiring trust between agents or the \smc. It does so by use of cryptographic methods and computer hardware that protect privacy while enabling the \smc to be audited.\end{itemize}

Finally, the smart contract delivers a limited but policy-relevant guarantee. If there exists at least one equilibrium that satisfies both broker-dealers’ hurdle spreads, then the protocol selects a hurdle-feasible trade and implements positive intermediation. This guarantee is not assured in the absence of commitment to an equilibrium-selection rule, and it is precisely the channel through which the protocol increases effective intermediation capacity under binding leverage constraints.

\subsection{Smart contract preliminaries}

The \bd{s} split surplus and are price-setters in their respective client markets. We modify equation \ref{eq:JPmaxpbm} to express the planner's objective function, which is to maximize constrained joint profit.

\begin{equation}\label{eq:sc_objective}
T^{\star} \in \arg\max_{T \in \mathcal{T}(\kappa)} \big(S^{-1}(T)-D^{-1}(T)\big)T,
\qquad
\mathcal{T}(\kappa) := \left\{T \in [0,\bar T] : \tfrac12\big(S^{-1}(T)-D^{-1}(T)\big) \ge \underline{\kappa} \right\}.
\tag{12}
\end{equation}

Here $\underline{\kappa} := \max\{\kappa_{mm},\kappa_{rm}\}$ is the binding hurdle that determines feasibility. By construction, if $\mathcal{T}(\kappa)$ is nonempty the smart contract selects a hurdle-feasible trade; if $\mathcal{T}(\kappa)$ is empty the protocol aborts and returns deposits, consistent with Step~5 of the protocol. Corollary \ref{corr:HurdleRateJPM} showed that such a \jpm configuration of trades exists. Our \smc achieves this outcome without imposing any informational requirements on any of the agents. Not only can it handle asymmetric information on client schedules, it achieves \jpm when the \bd{s} are ignorant of any feature of the counterparty's client supply or demand schedule for $T$. Notably, agents interacting with a \smc do not rely on trust in a individual or an organization, but rather they rely on the correct execution of the code.

\begin{definition}
A smart contract is a self-executing digital agreement with the terms of the agreement directly written into lines of code, stored on a single server or a blockchain. When predefined conditions are met, the contract automatically executes the corresponding actions, which can include moving objects between accounts.
\end{definition}

\textbf{State space}: Agents agree to trades on a  grid of \g-volumes $\mathcal{T} = \{T_{1},..,. T_{i},...,T_{I}\}$.  

\textbf{First-leg and second-leg} The \smc generates two things; (i) the \jpm tuple of repo rates and trade volume $\{r_{rm,*}, r_{mm,*}, T^{*}\}$ and (ii) the first and second-leg client prices, $\{p_{1,rm},p_{1,mm}$ and $\{p_{2,rm},p_{2,mm}\}$. This implies an \fl spread of $\frac{1}{2}(p_{1,rm} - p_{1,mm})T$ and similarly for the \sel. 

\textbf{Messages and objects} The \smc receives messages from \bd{s} containing client trade schedules and minimum spreads and sends \fl and \sel contracts to agents.  When traded objects \m and \g are appended to blockchains and held in agent wallets, the \smc operates an escrow that receives and sends \fl objects with agents. When objects are held as accounts at banks, the \fl escrow is held by the agent's bank. The banks message the \smc the receipt of objects in escrow and the \smc messages the banks with instructions for moving objects between banks and agent accounts.

\textbf{No collusion.} We place one restriction on agent behavior. We assume that client broker pairs trade at the repo rate reported to the \smc. This prohibits e.g. $BD_{mm}$ from reporting at volume $T$ an $r'_{mm}$ which is above the $r_{mm}$ on the client schedule. This could be implemented with a side payment where $BD_{mm}$ obtains from $MM$ a more favorable surplus split on $(r'_{mm} - r_{mm})T$ than it can obtain from $BD_{rm}$ (the area denoted by A in Figure \ref{fig:feasible_trade}).

\textbf{Client schedules} Each \bd{'s} client relationship pins down a reservation schedule mapping repo rates into quantities, which we take as fixed during the execution of the \smc. This schedule-taking is standard in dealer-market microstructure models that treat investors’ reservation values (or order-flow distributions) as primitives when characterizing dealers’ quote setting and intermediation.\footnote{See Duffie et.al. \cite{DuffieGarleanuPedersen2005} for a canonical OTC dealer-market model in which intermediaries bargain/quote to investors with heterogeneous reservation values which represents investor behavior via reduced-form demand objects.}

\subsection{Smart-contract protocol}
\label{subsec:smart_contract_protocol}

The \smc protocol works in the following sequential order.

\begin{enumerate}

    \item Each client- \bd pair agrees to a trade schedule $\{r_{i},T_{i}\}_{i=1}^{T_{max}}$ repo-rate: security pairs it is willing to trade, one for each value of $T$ (the \fl and \sel prices are agreed to and embedded in the repo rates).  
    
    \item  
    \begin{itemize}
    \item  Each $BD$ reports  the \smc its client trade schedule, signed by itself and its client. 
    \item The \smc selects one $BD$ to send its minimum spread of return (the ``first-mover''). 
    \item Each client sends its deposit to an escrow controlled by the \smc.\footnote{If the objects are tokenized, the escrow can be an account in another \smc, such as a stablecoin. If the objects are held in accounts at financial institutions, the escrow can be inside the institution. See Aronoff \cite{Aronoffnetwrap} for an explanation of how escrows can operate in this context.}
    \end{itemize}

    $\triangleright$ (Default) If a client fails to send its required object the contract is terminated and the client's deposit is distributed to the non-defaulting agents.

    $\triangleright$ The protocol is compatible with any formula for distributing the deposit.
    
    \item The \smc sends the first-mover's minimum spread to the other $BD$ (the ``second-mover'').

    \item The second-mover either accepts the first-mover's minimum spread, or reports a new minimum spread above the first-mover's minimum spread.

    $\triangleright$ The final minimum spread is the ``spread-constraint'' used by the \smc.
    
    \item the \smc computes the spread for each trade volume and discards any $T_{i}$ for which the spread is below the spread-constraint. 

    $\triangleright$ If all $\mathcal{T}$ are discarded the trade is aborted and deposits are returned.
    
    \item The \smc compares the profit at each remaining trade volume, $(r_{rm,i} - r_{mm,i})T_{i}$ and selects the trade volume $T^{*}$ that achieves \jpm.

    \item The \fl and \sel contracts between each client and its $BD$ automatically fix at the prices associated with $T^{*}$ from each client trade schedule, which construct $(r_{rm,*},T^{*})$ and $(r_{mm,*},T^{*})$ respectively. The interdealer repo rate is fixed at $r_{bd} = \frac{1}{2}(r_{rm,*} + r_{mm,*})$. 

    $\triangleright$ The \smc reports the contract terms to each agent, which is binding. 

    \item (Execution step): The \smc causes objects to be sent to each agent in accordance with the first-leg contracts.

    $\triangleright$  Any excess deposit is refunded to the clients.
        
\end{enumerate}

\subsection{Equilibrium in the smart contract game}
Here we prove existence and uniqueness of the constrained \jpm equilibrium under the \smc.

\subsubsection*{Deposits and margin}

The \smc is a sequence of steps. There is an interval of time between the commitment to contracts (Step 1) and the execution of trade (Step 5). During this gap a client might decide to defect. $MM$ could find a higher return and $RM$ could find a lower borrowing cost in other markets due to shocks that occur in the time interval.\footnote{A counter-argument is that the time interval may be so small as to practically eliminate the possibility of a new outside option appearing. In that case no deposit would be necessary.} The first instantiation of the deposit is designed to prevent this from occurring. 

\textbf{Client compliance at the \fl} is assured (or becomes irrelevant) if the maximum volume of objects are placed into escrow (``max deposit''). $RM$ sends $T_{I}$ into escrow and $MM$ sends the maximum \fl price, \m = $p_{1,mm,i}T_{i}$ on its client schedule. In that case there is no subgame where a \bd is required to send an object to meet its inter-dealer obligation because there are sufficient committed objects to complete any trade. \footnote{An initial deposit is required to ensure execution at the \fl because the legal penalty for defaulting on the contract before execution under e.g. the SIFMA master repurchase agreement (1996) paragraph 11 ``Default and Remedies'', is typically the counterparties lost profit, to the extent it cannot be mitigated. This is likely to be a small amount when the contract is negotiated in a competitive market. See e.g. the SIFMA master repurchase agreement (1996) paragraph 11 ``Events of Default'' \url{https://www.sifma.org/documents/master-repurchase-agreement-mra-2/}.}

\textbf{Compliance of all agents at the \sel} is assured by methods already in place in the repo market, namely haircuts and margin. The former is a subsidy on the \fl price to over-collateralize the loan made by $MM$ (in the form of the \fl purchase). Margin is an additional escrow payment made periodically when the time lag between the two legs is beyond one day. The basic idea of margin is to place the counterparties in a position where each is unaffected by a default by the other. For the current margin formula in the cleared US Treasuries repo market see \cite{ficc_gsd_clearing_fund}.

\subsubsection*{Reporting client schedules}

Here we state the proposition that $BD$'s will truthfully report their client schedules and we provide an informal proof. Appendix \ref{app:truthful_reporting} contains a rigorous proof. 

\begin{proposition*****}[Truthful reporting of client schedules]
\label{prop:client-scheds-smc}
Assuming commitments that bind $MM$ and $RM$ are in place, it is a dominant strategy for each $BD$ to truthfully report client schedules in the simultaneous-move game of the \smc protocol (Step 2).
\end{proposition*****}

% \begin{proof}[Proof sketch]
% Given the client authorization constraint, any feasible report by BD$_{mm}$ must satisfy $\hat r_{mm}(T)\ge r^c_{mm}(T)$ and any feasible report by BD$_{rm}$ must satisfy $\hat r_{rm}(T)\le r^c_{rm}(T)$ for all $T$ on the grid.
% Raising $\hat r_{mm}$ or lowering $\hat r_{rm}$ weakly reduces the smart-contract objective $(\hat r_{rm}(T)-\hat r_{mm}(T))T$ pointwise and can only shrink the admissible set under any minimum-spread filter,so it cannot improve a dealer's payoff (each receives a fixed share of the selected surplus). A full proof is in Appendix~E.
% \end{proof}

\begin{proof}[Proof sketch]
 We refer to the surplus-splitting rule for allocating fees $1/2(r_{rm} - r_{mm})T$, Figure \ref{fig:feasible_trade} below and the assumption of no collusion between \bd{s} and their clients. We first address $BD_{mm}$. There are two cases to consider; At some volume $T$, $BD_{mm}$ reports a client repo rate above the rate on the $MM$ client schedule, or $BD_{mm}$ reports a client repo rate below the rate on the $MM$ client schedule.   
\vskip5pt

\textbf{Report repo rate higher than the $MM$ schedule at some $T$:} In this case $BD_{mm}$ reports a client repo rate $r'_{mm}$ that is above the repo rate $r_{mm}$ on its client $MM$ schedule for $T$. The $MM$ will agree to deliver $T$ since its demand is increasing in $T$. The surplus splitting rule implies that $BD_{mm}$'s profit will be 1/2(A +B) when it reports the client schedule repo rate $r_{mm}$, and it will drop to 1/2(B) when it reports the higher repo rate $r'_{mm}$. In general, $BD_{mm}$'s profit from reporting a higher client repo rate will be lower \textit{for any client repo rate reported by $BD_{rm}$}. This proves that $BD_{mm}$ will never report a client repo rate above its client schedule at any $T$.

\textbf{Report repo rate lower than the client schedule at $T$:} This is ruled out by 'no -collusion'. The client is unwilling to fund the trade at $r''_{mm}$. so it will not report a trade at that volume. 

The conclusion is that it is a dominant strategy for $BD_{mm}$ to report the repo rate on its client schedule for all values of $T$. The same reasoning applies to $BD_{rm}$. 
\end{proof}

\begin{figure}[H]
\begin{center}
\includegraphics[page=1,width=0.5\textwidth,height = 0.3
\textheight]{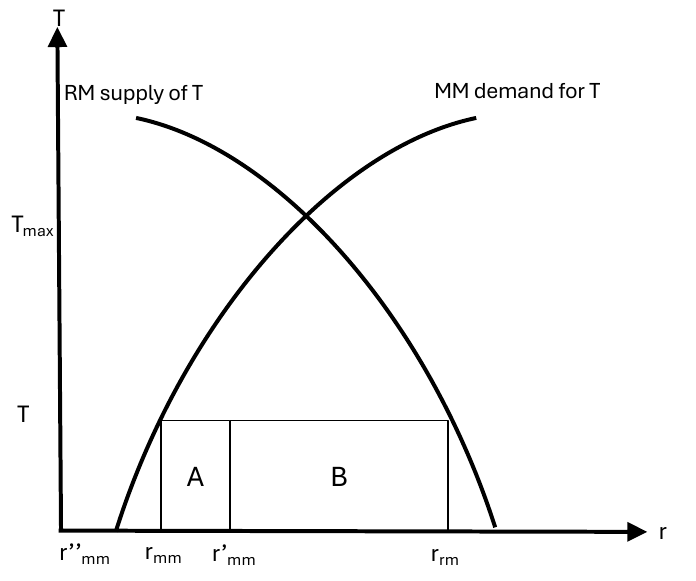}
\caption{Feasible Trade Graph}
\label{fig:feasible_trade}
\end{center}
\end{figure}

\subsubsection*{Reporting minimum spreads}

\begin{proposition******} (Reporting minimum spreads)
Let the dealers' true minimum per-dealer spreads be $\kappa_{mm}$ and $\kappa_{rm}$, and define $H:=\max\{\kappa_{mm},\kappa_{rm}\}$ and $L:=\min\{\kappa_{mm},\kappa_{rm}\}$. In Steps~3--4, the second mover selects the final spread-constraint $\hat\kappa$, subject to $\hat\kappa\ge$ the first
mover's report. Given $\hat\kappa$, the smart contract restricts attention to volumes $T$ satisfying $\tfrac12\!\left(r_{rm}(T)-r_{mm}(T)\right)\ge \hat\kappa$ and selects $T^*$ maximizing $\left(r_{rm}(T)-r_{mm}(T)\right)T$ on this restricted set.
Dealer $i\in\{mm,rm\}$ receives a negative payoff if
$\tfrac12\!\left(r_{rm}(T^*)-r_{mm}(T^*)\right)<\kappa_i$, and otherwise receives its profit share from the selected trade. Then truthful reporting is a weakly dominant strategy, and every equilibrium induces
$\hat\kappa=H$.\footnote{We assume a negative payoff will cause the \bd to exit.}
\end{proposition******}

\begin{proof}
Let true minimum spreads be $\kappa_{mm},\kappa_{rm}$ and define $H:=\max\{\kappa_{mm},\kappa_{rm}\}$, $L:=\min\{\kappa_{mm},\kappa_{rm}\}$.
The protocol sets the final spread-constraint $\hat\kappa$ equal to the second mover's report, with the constraint $\hat\kappa\ge$ (first mover's report). If the selected trade violates dealer $i$'s true minimum, $u_i=-\infty$; otherwise $u_i$ equals its profit share from the selected trade.

\emph{Second mover.} If the second mover's true minimum is $H$, any $\hat\kappa<H$ admits the possibility that the smart contract selects an infeasible trade for that dealer, yielding a negative payoff, so reporting at least $H$ is optimal.
Reporting $>H$ only shrinks the feasible set and weakly lowers profit, so the second mover weakly prefers $\hat\kappa=H$. If instead the second mover's true minimum is $L$, accepting $\hat\kappa=L$ is feasible and any report $>\max\{L,\text{(first report)}\}$
can only remove potentially profitable feasible trades, so accepting is weakly optimal.

\emph{First mover.} If the first mover's true minimum is $H$, reporting $<H$ risks that the second mover's true minimum is $L$ and accepts, resulting in $\hat\kappa<H$ and a possible $-\infty$ outcome; hence the first mover will not report below $H$. Reporting $>H$ can only reduce profit by restricting feasible trades. Thus truthful reporting is weakly dominant for the first mover.

Therefore, in equilibrium the final spread-constraint satisfies $\hat\kappa=H$ and truthful reporting is a weakly dominant strategy
for each dealer in this sequential reporting game.
\end{proof}

\begin{figure}[H]
\begin{center}
\includegraphics[page=1,width=0.5\textwidth,height = 0.3
\textheight]{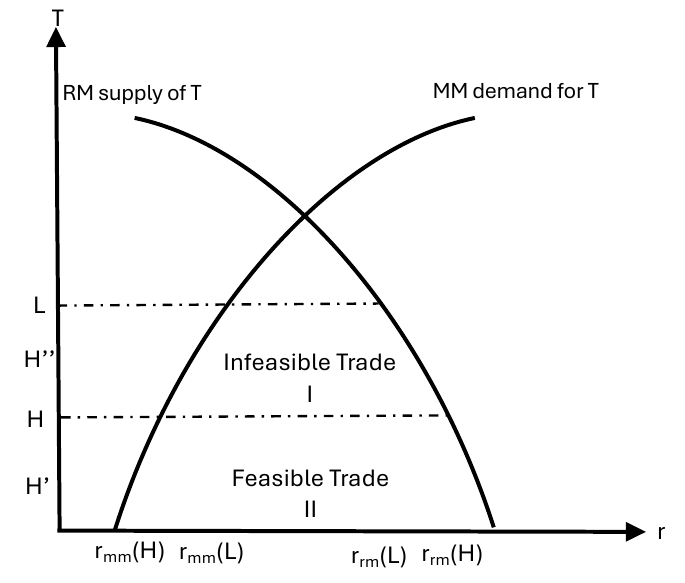}
\caption{Spread Constraint}
\label{fig:spread_constraint}
\end{center}
\end{figure}

A notable feature of these equilibrium results is that they can be implemented by myopic strategies where a $BD$'s reports its own client schedule and minimum spread, without need to consider any competitive response by its counterparty. Theorem 3 formalizes this implementability: in equilibrium each broker-dealer's report depends only on its own client schedule and its own hurdle, not on conjectures about the counterparty's reporting behavior.

\begin{theorem*****}[Truthful reporting equilibrium in the smart-contract game]
\label{thm:truthful_reporting_sc}
The reporting stage (Steps~2--4) admits a subgame-perfect equilibrium in which each broker-dealer reports its true client schedule in Step 2 and reports its true minimum half-spread in Steps~3--4. In this equilibrium, the final spread constraint satisfies
\[\hat\kappa = \max\{\kappa_{mm},\kappa_{rm}\},\]
and the smart contract implements the constrained joint-profit-maximizing trade selected by its objective (equation (12)) subject to the feasibility restriction induced by $\hat\kappa$.
\end{theorem*****}

\begin{proof}
By Proposition 6, truthful reporting of client schedules is a dominant strategy in Step~2. By Proposition 7, truthful reporting of minimum spreads is a weakly dominant strategy in Steps~3--4 and every equilibrium of that subgame induces $\hat\kappa=\max\{\kappa_{mm},\kappa_{rm}\}$. Combining these facts yields a subgame-perfect equilibrium of Steps~2--4 with truthful reports. The implemented trade then follows directly from the smart contract's deterministic selection rule given the reports.
\end{proof}

\subsection{Integrity, privacy and auditability}
\label{subsec:Privacy_auditability}

This section formalizes the cryptographic guarantees delivered by the \smc described above. Two broker–dealers, $BD_{r_{rm}}$ and $BD_{r_{mm}}$, privately report schedules of trades. The \smc computes and delivers to each BD the single trade $\{r_{rm,*}, r_{mm,*},T^{*}\}$ that maximizes joint profit under the pre‑programmed objective in (equation \ref{eq:sc_objective}), while revealing nothing else about either schedule. Correctness is publicly auditable via a succinct zero‑knowledge proof (zk‑SNARK/NIZK), and liveness/fairness are enforced by hash‑timelocks with cancellation on timeout. These guarantees are obtained in either of two implementations: (i) a generic hardware protected trusted‑execution environment (TEE) and (ii) a public‑chain smart contract that verifies zero‑knowledge proofs on‑chain. In both cases the computation is performed on a server with a TEE. Figure \ref{fig:TEE} depicts the encrypted incoming messages, the decryption and execution inside the TEE, the recording of the transcript, the re-encryption and sending of the computed output to the parties.

\begin{figure}[H]
\begin{center}
\includegraphics[page=1,width=0.5\textwidth,height = 0.3\textheight]{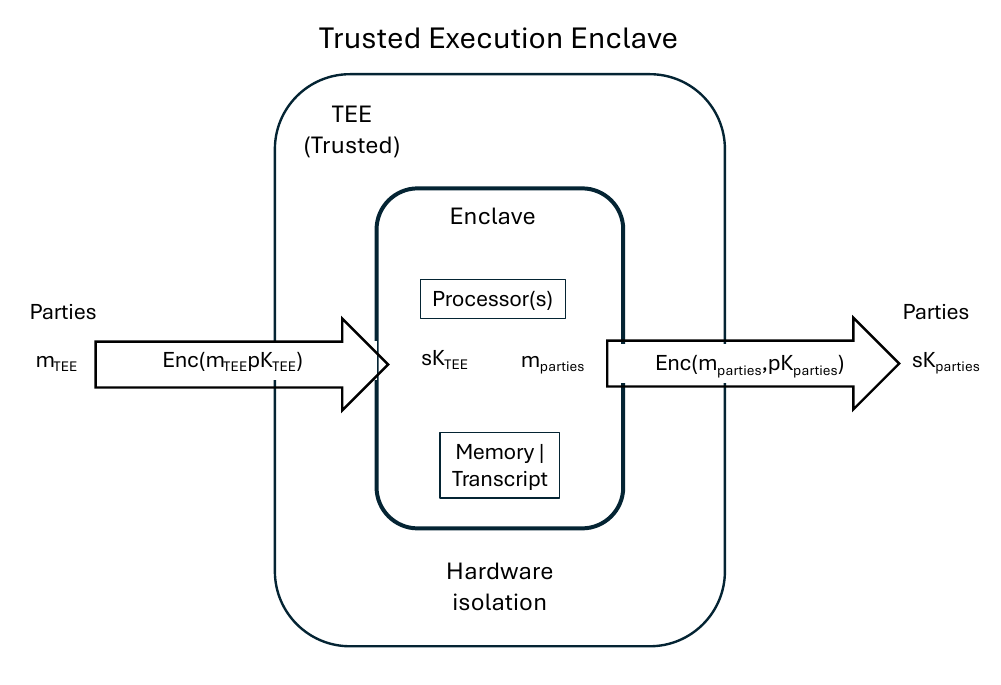}
\end{center}
\caption{Trusted Execution Enclave}
\label{fig:TEE}
\end{figure}

\paragraph{Cryptographic primitives (both implementations).}
Each $BD$ forms a binding, hiding vector commitment to its schedule (e.g., a Merkle/Pedersen root): $C_{rm}$ commits to $S_{r_{rm}}$ and $C_{mm}$ to $D_{r_{mm}}$. \footnote{Catalano and Fiori \cite{catalanoFiore2013vc}, Merkle \cite{merkle1989certified} and Pederson \cite{pedersen1992vss} describe ector commitments.} A hash‑timelock establishes phased deadlines: \textsc{Commit}~$\to$~\textsc{Prove/Verify}~$\to$~\textsc{Deliver}. If a deadline is missed, the protocol aborts without disclosure. The \smc must output only $\{r_{rm,*},r_{mm,*},T^{*}\}$ (encrypted separately to each $BD$’s public key), together with a publicly verifiable proof that $\{r_{rm,*},r_{mm,*},T^{*}\}$ is admissible and profit‑maximizing with respect to the posted commitments and the fixed code implementing the \smc protocol.\footnote{See Todd \cite{bip65} and Friedenbach and Lombrozo \cite{bip112}.}

\textit{Zero‑knowledge statement (high level).} The proof attests that: (i) $T^{*}\in S_{r_{rm}}$ and $T^{*}\in D_{r_{mm}}$ and satisfies equation \ref{eq:sc_objective} (via commitment ‑ membership proofs against $C_{rm},C_{mm}$); (ii) $\kappa(T^{*}) \ge \kappa(T)$ for every admissible candidate $T$ encoded by the bounded schedule grid $\mathbb{T}$; and (iii) the published ciphertexts $E_{rm}=ENC_{pk_{rm}}(T^{*})$ and $E_{mm}=ENC_{pk_{mm}}(T^{*})$ are correct encryptions of the same $T^{*}$. The ground truth is the transcript inside the TEE. No plaintext schedule elements are revealed.\footnote{Thaler \cite{ThalerZKP} is a standard reference for ZKP methods.}

\paragraph{Outcomes}

\begin{itemize}

\item \textit{Integrity.} The only accepted outcome is one for which a valid proof links $\{r_{rm,*},r_{mm,*},T^{*}\}$ to the committed schedules and the fixed objective (equation \ref{eq:sc_objective}); thus the selected trade is the joint‑profit maximizer encoded by the program. 

\item \textit{Privacy.} Beyond the delivery of the actual trade, $\{r_{rm,*},r_{mm,*},T^{*}\}$ (privately, as $E_{rm},E_{mm}$), no information about $S_{r_{rm}}$ or $D_{r_{mm}}$is revealed; the proof is zero‑knowledge. 

\item \textit{Auditability.} Verification of $\kappa$ (and of attestation in the TEE case) lets any $BD$—or third party—confirm correctness without seeing inputs. 

\item \textit{Fairness/liveness.} Hash‑timelocks deliver ordering and bounded execution; missed deadlines trigger cancellation without compromising privacy. 
\end{itemize}

Appendix \ref{app:crypto_implementation} contains the implementation details.

% === END: The_Smart_Contract_v4.tex ===

% === INLINED: Conclusion-2.tex ===
\section{Conclusion}
\label{sec:Conclusion}

We presented a model of intermediated trade in the U.S. Treasuries repo market. We focused on a repo chain composed of two \bd{s} flanked on either side by, respectively, an ultimate repo borrower and an ultimate repo lender. Each element of our model - from the length of the repo chain to the timing of trading to the functional forms of supply and demand schedules to the bargaining power of agents - is derived from the empirical and theoretical literature. We demonstrated that strategic interactions between agents leads to multiple equilibrium. 

We then introduced a \smc protocol that resolves the coordination and hurdle challenge by selecting the \jpm volume of trade as an equilibrium, which is implemented by a dominant strategy of broker-dealer truthfully reporting its client schedule. The \smc also selects a spread-constraint (between borrowing and lending rates) , which bounds the region of feasible trades. This constraint reflects the shadow-price of allocating balance-sheet capacity to the broker-dealers. We demonstrate that the optimal spread-constraint can be implemented by a weakly dominant strategy of broker-dealer truthfully reporting its true hurdle-constraint. 

The \smc is a self-executing digital agreement with the terms of the agreement embedded in the code. Agent trust resides in the correctness of the code  -- which can be audited --- rather than the conduct of an individual or entity.  Broker-dealers send messages containing contingent trading schedules - repo rates associated with each volume of trade -  with their clients and their internal minimum spread of return. the \smc selects the trade that maximizes \bd joint profit, which is an equilibrium. Agent compliance with the rules are enforced by deposits posted by the clients. 

% We describe a design choice that results from a tradeoff between participation in the \smc, which is decreasing in the size of the required deposit, and compliance with the protocol, which is increasing in the size of the deposit. 

Privacy is maintained by encrypting messages between agents and the \smc and protecting leakage from the computations. We explored two implementations. In one the computation and auditing takes place inside secure hardware. In the other the encrypted output is appended to a public blockchain for auditing. In both cases privacy of the auditing is achieved by using zero knowledge proofs.   

The \smc has practical significance. 

\begin{itemize}

\item It provides a policy-relevant floor on intermediation under leverage constraints. The protocol does not generally maximize repo volume. Instead, it guarantees that whenever there exists at least one equilibrium satisfying both broker-dealers’ hurdle spreads, the implemented outcome is hurdle-feasible and features positive intermediation, thereby reducing the risk that equilibrium-selection failures translate binding leverage regulation into a collapse to no trade.

\item It can be implemented by the myopic strategy of truthful reporting, which minimizes the cognitive burden on agents by eliminating the need to assess knowledge of counterparties and to simulate counterparty responses. 
\end{itemize}

Finally, this paper contributes to the understanding of how cryptography, hardware and \smc technologies can be integrated into market designs to overcome obstacles of trust and coordination and improve economic outcomes.
% === END: Conclusion-2.tex ===

%\bibliographystyle{plain}
% === INLINED BIBLIOGRAPHY (.bbl) ===

% === END INLINED BIBLIOGRAPHY ===

\appendix
%\input{App_model of concave functions}
% === INLINED: app_excess_inventories.tex ===
\section{Excess Inventories}
\label{app:excess_inventories}

\begin{proposition*}[No excess inventory in \idt Equilibrium]
\label{corr:no excess inventory in equilibrium}
$D_{r_{mm}} = S_{r_{rm}}$ in equilibrium, i.e. neither $BD$ will carry excess inventories into the interdealer trade.
\end{proposition*}

\begin{proof}
Fix $r_{rm}$. If $D(r_{mm})>S(r_{rm})$, then $Q(r_{mm},r_{rm})=S(r_{rm})$ is locally constant in $r_{mm}$, while BD$_{mm}$ has unsold inventory $D(r_{mm})-Q>0$. Decreasing $r_{mm}$ slightly (still keeping $D(r_{mm})>S(r_{rm})$) strictly increases the spread $(r_{rm}-r_{mm})$ and weakly reduces unsold inventory, so BD$_{mm}$'s payoff strictly increases because the inventory penalty is nondecreasing. Hence any best response of BD$_{mm}$ must satisfy $D(r_{mm})\le S(r_{rm})$.

Symmetrically, fixing $r_{mm}$, if $S(r_{rm})>D(r_{mm})$ then increasing $r_{rm}$ slightly raises the spread and weakly reduces BD$_{rm}$'s unsold inventory, so no best response of BD$_{rm}$ can satisfy $S(r_{rm})>D(r_{mm})$. Thus any best response of BD$_{rm}$ satisfies $S(r_{rm})\le D(r_{mm})$.

In a Nash equilibrium $(r_{mm}^*,r_{rm}^*)$, both inequalities hold, hence
$D(r_{mm}^*)=S(r_{rm}^*)$.
\end{proof}

\section{Funding Liquidity}
\label{app:funding_liquidity}

We model a funding commitment as a restriction on a broker-dealer's feasible client quote: the
dealer must choose a client repo rate that induces at least a minimum client-side quantity. This pins down a lower (resp.\ upper) bound on the dealer's quote because client schedules are monotone.

\begin{definition}[Funding commitment]\label{def:funding_commitment}
Fix minimum committed quantities $T^{mm}\ge 0$ and $T^{rm}\ge 0$.
A \textbf{funding commitment} requires the client-side quantities induced by the dealers' quotes to satisfy
\[D(r_{mm}) \ge T^{mm}, \qquad S(r_{rm}) \ge T^{rm}.\]
Equivalently (since $D$ is strictly increasing and $S$ is strictly decreasing),
\[r_{mm} \ge r^{fc}_{mm} := D^{-1}(T^{mm}), \qquad r_{rm} \le r^{fc}_{rm} := S^{-1}(T^{rm}),\]
with the convention that $T^{mm}=0$ gives $r^{fc}_{mm}=0$ and $T^{rm}=0$ gives $r^{fc}_{rm}=r_b$.
\end{definition}

As in the main text, the interdealer trade quantity is
\[Q(r_{mm},r_{rm}) \;:=\; \min\{D(r_{mm}),S(r_{rm})\},\]
and dealers may carry excess inventories when client commitments force $D(r_{mm})\neq S(r_{rm})$. We assume each dealer suffers a (weakly) increasing penalty in its own unsold inventory, as in the Inventory Penalties assumption used to prove Proposition~3.1.

\begin{proposition*******}[Existence of equilibrium with client funding commitments]\label{prop:funding_eqm} There is at least one Nash equilibrium in client repo rates $(r_{mm},r_{rm})$ when one or both broker-dealersis subject to a funding commitment (Definition~\ref{def:funding_commitment}).
\end{proposition*******}

\begin{proof}
The feasible sets $R_{mm}=[r_{mm}^{fc},r_b]$ and $R_{rm}=[0,r_{rm}^{fc}]$ are nonempty, compact, and convex. Payoffs are continuous in $(r_{mm},r_{rm})$ because $D$, $S$, $Q=\min\{D,S\}$, and the penalty functions are continuous.

Fix $r_{rm}$. On the region $\{r_{mm}:D(r_{mm})\le S(r_{rm})\}$ we have $Q=D(r_{mm})$ and zero penalty, so
$\pi_{mm}(r_{mm},r_{rm})=\tfrac12(r_{rm}-r_{mm})D(r_{mm})$ is concave in $r_{mm}$.
On the region $\{r_{mm}:D(r_{mm})> S(r_{rm})\}$, $Q=S(r_{rm})$ is constant, the spread term strictly decreases in $r_{mm}$, and the unsold-inventory penalty weakly increases, so $\pi_{mm}$ is strictly decreasing there. Therefore $\pi_{mm}(\cdot,r_{rm})$ is quasi-concave on $R_{mm}$ and attains a maximum; hence
$BR_{mm}(r_{rm})$ is a nonempty compact interval (thus convex). By symmetry, $BR_{rm}(r_{mm})$ is a nonempty compact interval for each $r_{mm}$.

By Berge's Maximum Theorem, both best-response correspondences are upper hemicontinuous. Hence $BR(r_{mm},r_{rm})=BR_{mm}(r_{rm})\times BR_{rm}(r_{mm})$ maps the compact convex set $R_{mm}\times R_{rm}$ into nonempty compact convex values and is upper hemicontinuous. Kakutani's theorem yields a fixed point, i.e.\ a Nash equilibrium in client repo rates.
\end{proof}

\begin{figure}[H]
\begin{center}
\includegraphics[page=1,width=0.55\textwidth,height = 0.3
\textheight]{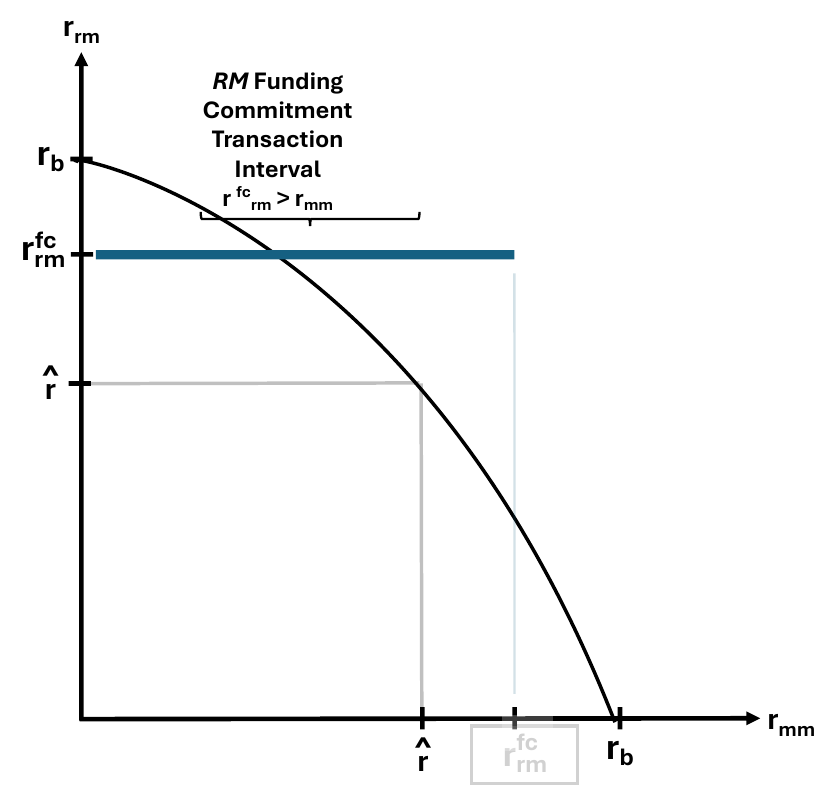}
\caption{Funding Liquidity $BD_{rm}$ lock}
\label{fig:Funding Liquidity_BDrm lock}
\end{center}
\end{figure}

\begin{figure}[H]
\begin{center}
\includegraphics[page=1,width=0.55\textwidth,height = 0.3
\textheight]{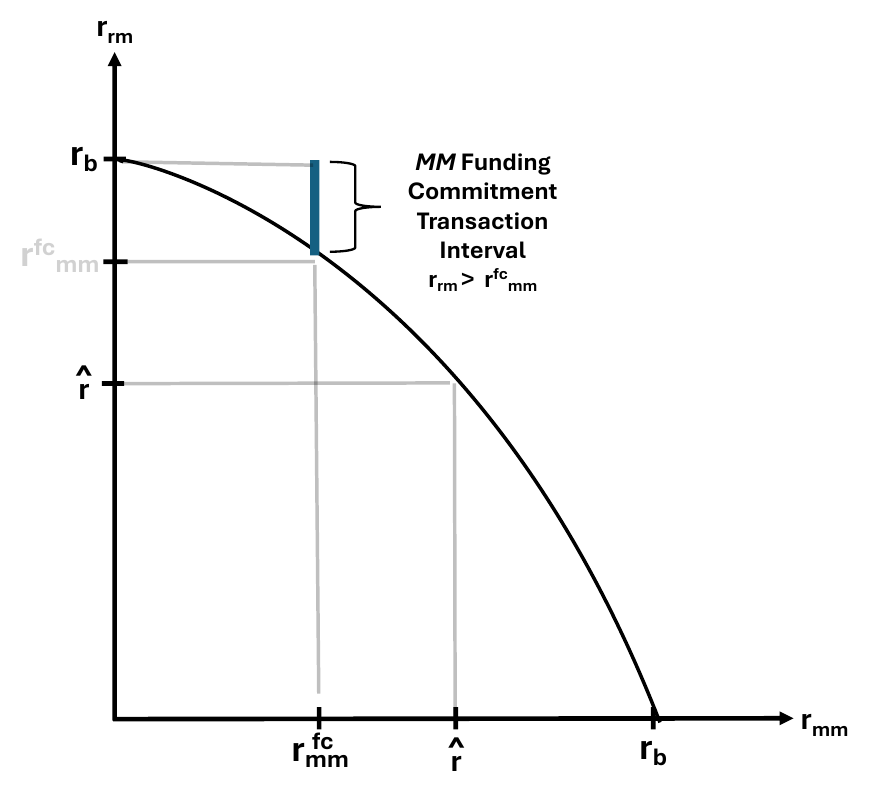}
\caption{Funding Liquidity $BD_{mm}$ lock}
\label{fig:Funding Liquidity_BDmm lock}
\end{center}
\end{figure}

A key takeaway from Proposition \ref{prop:funding_eqm} is that a \bd may carry excess inventories in equilibrium due to a commitment to provide funding liquidity to its client. In the case of a commitment to $RM$, $BD_{rm}$ will carry excess inventories if $r_{mm}$ lies to the left of the intersection of the $RM$ funding commitment interval with the balance transaction surface. In the case of a commitment to $MM$, $BD_{mm}$ will carry excess inventories if $r_{rm}$ lies below the intersection of the $MM$ funding commitment interval with the balance transaction surface.

\section{Proof of Existence of Joint Profit Maximization Equilibrium}
\label{app:Proof_of_JPM_Equilibrium}

\begin{proposition****}[$BD$ Joint Profit Maximization Repo Rates -  Existence]
\label{prop: $BD$ Joint Profit Maximization Repo Rates -  Existence}
There is a  repo rate pair which attains $BD$ joint profit maximization. 
\end{proposition****}

\begin{proof}
Define the feasible trade interval $[0,T_{\max}]$, where $T_{\max}$ is the maximal balanced trade volume
(on the balanced-trade surface), and define the joint intermediation surplus as a function of trade volume:
\[\pi(T) := \bigl(S^{-1}(T)-D^{-1}(T)\bigr)\,T,\qquad T\in[0,T_{\max}].\]
Because $D(\cdot)$ is continuous and strictly increasing and $S(\cdot)$ is continuous and strictly decreasing,
their inverses $D^{-1}$ and $S^{-1}$ exist on $[0,T_{\max}]$ and are continuous there. Hence $\pi(\cdot)$ is continuous on the compact set $[0,T_{\max}]$.

By the Weierstrass Extreme Value Theorem, there exists $T^\ast\in[0,T_{\max}]$ such that
\[\pi(T^\ast) = \max_{T\in[0,T_{\max}]}\pi(T).\]
Let $r_{mm}^\ast := D^{-1}(T^\ast)$ and $r_{rm}^\ast := S^{-1}(T^\ast)$. Then by construction
$D(r_{mm}^\ast)=S(r_{rm}^\ast)=T^\ast$, and the pair $(r_{mm}^\ast,r_{rm}^\ast)$ attains joint broker-dealer profit maximization among balanced allocations. This proves existence.
\end{proof}

\begin{corollary*}[$BD$ Joint Profit Maximization Equilibrium]
\label{corr: $BD$ Joint Profit Maximization Repo Rates-Eqm}
There is a joint profit maximizing equilibrium.
\end{corollary*}

\begin{proof}
Let $T^*\in[0,T_{\max}]$ maximize $\Pi(T):=(S^{-1}(T)-D^{-1}(T))T$ (Proposition~5), and set
$r_{mm}^*:=D^{-1}(T^*)$, $r_{rm}^*:=S^{-1}(T^*)$. Then trade is balanced:
$D(r_{mm}^*)=S(r_{rm}^*)=T^*$.

First consider BD$_{mm}$ holding $r_{rm}=r_{rm}^*$ fixed. Any deviation to $r'_{mm}>r_{mm}^*$ implies $D(r'_{mm})>T^*$ and creates unsold inventory; with nondecreasing inventory penalties this cannot be profitable (same deviation logic as in Proposition~1). For $r_{mm}<r_{mm}^*$ there is no inventory penalty and
BD$_{mm}$'s payoff is proportional to $\phi(r_{mm};r_{rm}^*)=\tfrac12(r_{rm}^*-r_{mm})D(r_{mm})$, whose left-derivative at $r_{mm}^*$ equals
\[\frac12\Big((r_{rm}^*-r_{mm}^*)D'(r_{mm}^*)-T^*\Big).\]
If $T^*$ is interior, the FOC $\Pi'(T^*)=0$ implies
\[r_{rm}^*-r_{mm}^* \;=\; T^*\Big((D^{-1})'(T^*)-(S^{-1})'(T^*)\Big).\]
Using $D'(r_{mm}^*)=1/(D^{-1})'(T^*)$ gives
\[\frac{\partial \phi}{\partial r_{mm}}(r_{mm}^*-) \;=\;
-\frac12\,T^*\,\frac{(S^{-1})'(T^*)}{(D^{-1})'(T^*)} \;>\;0\]
since $(D^{-1})'(T^*)>0$ and $(S^{-1})'(T^*)<0$. Hence lowering $r_{mm}$ reduces BD$_{mm}$'s payoff.

Next consider BD$_{rm}$ holding $r_{mm}=r_{mm}^*$ fixed. Any deviation to $r'_{rm}<r_{rm}^*$ implies $S(r'_{rm})>T^*$ and creates unsold inventory, so it is not profitable. For $r_{rm}>r_{rm}^*$ there is no inventory penalty and payoff is proportional to $\psi(r_{rm};r_{mm}^*)=\tfrac12(r_{rm}-r_{mm}^*)S(r_{rm})$,
with right-derivative at $r_{rm}^*$ equal to
\[\frac12\Big(T^*+(r_{rm}^*-r_{mm}^*)S'(r_{rm}^*)\Big)
= \frac12\,T^*\,\frac{(D^{-1})'(T^*)}{(S^{-1})'(T^*)} \;<\;0,\]
so increasing $r_{rm}$ reduces BD$_{rm}$'s payoff.

Thus $(r_{mm}^*,r_{rm}^*)$ is a constrained outcome and hence a Nash equilibrium by Proposition~2. (If $T^*$ is at an endpoint, the same conclusion follows from one-sided derivatives and the inventory-penalty
argument.)
\end{proof}

\section{Proof of Multiple Equilibrium Under Asymmetric Information}
\label{app:asymmetric_eqm}

This appendix provides a  proof of existence and multiplicity of equilibrium under
asymmetric information. The main idea is to construct volume-targeting Bayesian Nash equilibria: each broker-dealer chooses its client repo rate so that, given its privately observed client state, it transacts a fixed target inventory $T$ with its client. We show that for a nondegenerate set of target volumes $T$, these strategies are best responses for every type (hence no beliefs about the counterparty's type are required).

\subsection{Bayesian environment and equilibrium concept}\label{app:C1}

There are two broker-dealers, BD$_{mm}$ and BD$_{rm}$, and two client-side schedules:

\begin{itemize}
\item Money-market (MM) demand for collateral (equivalently, demand to lend cash): 
$D(\,r_{mm},\theta_{mm}\,)$, where $r_{mm}$ is the client repo rate set by BD$_{mm}$ and
$\theta_{mm}$ is a privately observed MM state.
\item Repo-market (RM) supply of collateral (equivalently, willingness to borrow cash):
$S(\,r_{rm},\theta_{rm}\,)$, where $r_{rm}$ is the client repo rate set by BD$_{rm}$ and
$\theta_{rm}$ is a privately observed RM state.
\end{itemize}

Types satisfy $\theta_{mm}\in[\underline\theta_{mm},\bar\theta_{mm}]$ and
$\theta_{rm}\in[\underline\theta_{rm},\bar\theta_{rm}]$ with a common-knowledge joint distribution. BD$_{mm}$ observes $\theta_{mm}$ only, BD$_{rm}$ observes $\theta_{rm}$ only.

Given chosen rates $(r_{mm},r_{rm})$ and types $(\theta_{mm},\theta_{rm})$, the interdealer traded
quantity is
\[Q \;:=\; \min\{D(r_{mm},\theta_{mm}),\,S(r_{rm},\theta_{rm})\}.\]
We maintain the equal surplus split: each dealer's trading profit is proportional to
$\frac12 (r_{rm}-r_{mm})Q$. We allow (as in the main text) for an inventory penalty that makes it weakly undesirable to choose a client rate that generates unsold inventory:
BD$_{mm}$ is penalized when $D(r_{mm},\theta_{mm})>Q$, and BD$_{rm}$ is penalized when
$S(r_{rm},\theta_{rm})>Q$.

\begin{definition}[Bayesian Nash equilibrium]\label{def:BNE}
A pair of measurable strategies
\[\sigma_{mm}:\theta_{mm}\mapsto r_{mm},\qquad \sigma_{rm}:\theta_{rm}\mapsto r_{rm}\]
is a Bayesian Nash equilibrium if for each type $\theta_{mm}$, $\sigma_{mm}(\theta_{mm})$ maximizes BD$_{mm}$'s expected payoff given $\sigma_{rm}$, and for each type $\theta_{rm}$,
$\sigma_{rm}(\theta_{rm})$ maximizes BD$_{rm}$'s expected payoff given $\sigma_{mm}$.
\end{definition}

\subsection{Regularity and the type-shift structure}

We impose standard monotonicity/concavity in own rate and a convenient (but empirically natural) type shift structure that makes the equilibrium construction clean.

\begin{assumption}[Schedule regularity]\label{ass:regularity}
For each $\theta_{mm}$, $D(\cdot,\theta_{mm})$ is $C^2$, strictly increasing, and strictly concave on its domain, with $D(0,\theta_{mm})=0$.
For each $\theta_{rm}$, $S(\cdot,\theta_{rm})$ is $C^2$, strictly decreasing, and strictly concave on its domain, with $S(r_b,\theta_{rm})=0$. Types shift schedules outward: $\partial_{\theta_{mm}}D(r,\theta_{mm})>0$ and $\partial_{\theta_{rm}}S(r,\theta_{rm})>0$ for all relevant $(r,\theta)$.
\end{assumption}

\begin{assumption}[Multiplicative type shocks]\label{ass:multiplicative}
There exist baseline schedules $\bar D$ and $\bar S$ and positive scalars $\theta_{mm},\theta_{rm}$ such that
\[D(r_{mm},\theta_{mm})=\theta_{mm}\,\bar D(r_{mm}),\qquad
S(r_{rm},\theta_{rm})=\theta_{rm}\,\bar S(r_{rm}),\]
where $\bar D$ is strictly increasing and strictly concave and $\bar S$ is strictly decreasing and strictly concave (and satisfies $\bar D(0)=0$ and $\bar S(r_b)=0$).
\end{assumption}

Assumption~\ref{ass:multiplicative} is a sufficient condition that yields tractable comparative statics. It captures the common case where client ``states'' scale the entire schedule (e.g.\ demand/supply shifters).

\begin{lemma}[Type-contingent inverse rates]\label{lem:inverse-rates}
Fix any target volume $T\in[0,T_{\max}]$ where $T_{\max}:=\underline\theta_{mm}\bar D(\hat r) =\underline\theta_{rm}\bar S(\hat r)$ (the maximal balanced trade at the minimal-type profile). Under Assumption~\ref{ass:multiplicative}, for each $\theta_{mm}$ and $\theta_{rm}$ there exist unique rates
\[r_{mm}(T,\theta_{mm}) := \bar D^{-1}\!\Bigl(\frac{T}{\theta_{mm}}\Bigr),\qquad
r_{rm}(T,\theta_{rm}) := \bar S^{-1}\!\Bigl(\frac{T}{\theta_{rm}}\Bigr),\]
such that $D(r_{mm}(T,\theta_{mm}),\theta_{mm})=T$ and $S(r_{rm}(T,\theta_{rm}),\theta_{rm})=T$. Moreover, $r_{mm}(T,\theta_{mm})$ is strictly decreasing in $\theta_{mm}$ and
$r_{rm}(T,\theta_{rm})$ is strictly increasing in $\theta_{rm}$.
\end{lemma}

\begin{proof}
Under Assumption~2, $D(\cdot,\theta_{mm})=\theta_{mm}\bar D(\cdot)$ with $\bar D$ strictly increasing, so $r_{mm}(T,\theta_{mm})=\bar D^{-1}(T/\theta_{mm})$ is well-defined and unique; similarly $r_{rm}(T,\theta_{rm})=\bar S^{-1}(T/\theta_{rm})$. Monotonicity in types follows since $T/\theta$ decreases in $\theta$ and $\bar D^{-1}$ is increasing while $\bar S^{-1}$ is decreasing.
\end{proof}

% \begin{proof}
% Existence and uniqueness follow because $\bar D$ is strictly increasing and continuous and
% $\bar S$ is strictly decreasing and continuous, hence both inverses exist on the relevant range.
% Monotonicity in types follows since $T/\theta_{mm}$ decreases in $\theta_{mm}$ and $\bar D^{-1}$ is increasing, while $T/\theta_{rm}$ decreases in $\theta_{rm}$ and $\bar S^{-1}$ is decreasing (because $\bar S$ is decreasing).
% \end{proof}

\subsection{Target-volume strategies and best responses}

Fix a target volume $T\in[0,T_{\max}]$. Define the $T$-targeting strategies:
\begin{equation}\label{eq:T-target-strats}
\sigma^T_{mm}(\theta_{mm}) := r_{mm}(T,\theta_{mm}),\qquad\sigma^T_{rm}(\theta_{rm}) := r_{rm}(T,\theta_{rm}),
\end{equation}
so that each dealer induces its client-side quantity equal to $T$ for every type realization.
We now show that these strategies constitute a Bayesian Nash equilibrium whenever $T$ corresponds to a (constrained) complete-information equilibrium at the minimal-type profile.

\begin{lemma}[Type invariance of unconstrained peaks under multiplicative shocks]\label{lem:peak-invariance}
Under Assumption~\ref{ass:multiplicative}, in the no-inventory region the unconstrained maximizers
(the ``peaks'') of BD$_{mm}$ and BD$_{rm}$ do not depend on their own types. More precisely: for any fixed counterparty rate $r_{rm}$, BD$_{mm}$'s no-inventory objective is proportional to
\[(r_{rm}-r_{mm})\,\bar D(r_{mm}),\]
so its unique maximizer $r_{mm}^{A}(r_{rm})$ is independent of $\theta_{mm}$.
For any fixed counterparty rate $r_{mm}$, BD$_{rm}$'s no-inventory objective is proportional to
\[(r_{rm}-r_{mm})\,\bar S(r_{rm}),\]
so its unique maximizer $r_{rm}^{A}(r_{mm})$ is independent of $\theta_{rm}$.
\end{lemma}

\begin{proof}
In the no-inventory region, $Q=\theta_{mm}\bar D(r_{mm})$ for BD$_{mm}$, so its payoff is a positive scalar $\theta_{mm}$ times $(r_{rm}-r_{mm})\bar D(r_{mm})$; multiplying by $\theta_{mm}>0$ does not change the argmax. The BD$_{rm}$ case is identical with $Q=\theta_{rm}\bar S(r_{rm})$.
\end{proof}

% \begin{proof}
% For BD$_{mm}$ in the no-inventory region, $Q=D(r_{mm},\theta_{mm})=\theta_{mm}\bar D(r_{mm})$, so (up to the constant factor $\frac12$) the payoff is
% \[(r_{rm}-r_{mm})\,\theta_{mm}\bar D(r_{mm}) = \theta_{mm}\cdot\bigl((r_{rm}-r_{mm})\bar D(r_{mm})\bigr).\]
% Multiplying an objective by a positive constant does not change its argmax; hence the maximizer is independent of $\theta_{mm}$. The BD$_{rm}$ case is identical.
% \end{proof}

\begin{lemma}[Dominant-response condition for $T$-targeting]\label{lem:dominant-response}
Fix $T\in(0,T_{\max}]$ and suppose BD$_{rm}$ plays $\sigma^T_{rm}$.
For any type $\theta_{mm}$, BD$_{mm}$'s best response is $\sigma^T_{mm}(\theta_{mm})$ provided that
\begin{equation}\label{eq:BR-mm-condition}
r_{mm}(T,\theta_{mm}) \;\le\; \min_{\theta_{rm}\in[\underline\theta_{rm},\bar\theta_{rm}]}\,
r_{mm}^{A}\!\bigl(r_{rm}(T,\theta_{rm})\bigr).
\end{equation}
Symmetrically, if BD$_{mm}$ plays $\sigma^T_{mm}$, then for any type $\theta_{rm}$, BD$_{rm}$'s best response is
$\sigma^T_{rm}(\theta_{rm})$ provided that

\begin{equation}\label{eq:BR-rm-condition}
r_{rm}(T,\theta_{rm}) \;\ge\; \max_{\theta_{mm}\in[\underline\theta_{mm},\bar\theta_{mm}]}\,
r_{rm}^{A}\!\bigl(r_{mm}(T,\theta_{mm})\bigr).
\end{equation}
\end{lemma}

\begin{proof}
We prove the BD$_{mm}$ statement; the BD$_{rm}$ statement is symmetric.

Fix $T$ and type $\theta_{mm}$, and suppose BD$_{rm}$ plays $\sigma^T_{rm}$, so $S(r_{rm}(T,\theta_{rm}),\theta_{rm})=T$ for every $\theta_{rm}$. If BD$_{mm}$ chooses $r_{mm}$ with $D(r_{mm},\theta_{mm})>T$, then for every $\theta_{rm}$ the traded quantity is capped at $Q=T$ and BD$_{mm}$ carries unsold inventory, which is (weakly) costly. Hence any best response satisfies
$D(r_{mm},\theta_{mm})\le T$, i.e.\ $r_{mm}\le r_{mm}(T,\theta_{mm})$.

On the feasible set $\{r_{mm}:r_{mm}\le r_{mm}(T,\theta_{mm})\}$ there is no inventory penalty and, for each $\theta_{rm}$, the no-inventory objective is single-peaked with peak at $r_{mm}^A(r_{rm}(T,\theta_{rm}))$. Condition~(16) implies
$r_{mm}(T,\theta_{mm})$ lies weakly to the left of every such peak, so the objective is weakly increasing up to the boundary. Thus the best response is attained at $r_{mm}(T,\theta_{mm})$, i.e.\ $\sigma^T_{mm}(\theta_{mm})$.
\end{proof}

% \begin{proof}
% We prove the BD$_{mm}$ statement; the BD$_{rm}$ statement is analogous.

% Fix $\theta_{mm}$ and let BD$_{rm}$ play $\sigma^T_{rm}$. Then for every $\theta_{rm}$,
% $S(r_{rm}(T,\theta_{rm}),\theta_{rm})=T$. Hence if BD$_{mm}$ chooses any rate with
% $D(r_{mm},\theta_{mm})>T$, the traded quantity is capped at $Q=T$ for \emph{all} counterparty types and BD$_{mm}$ carries unsold inventory. Under the inventory-penalty logic (as in Proposition~3.1), such deviations are not profitable relative to moving back to the boundary $D=T$.

% Therefore, any best response must satisfy $D(r_{mm},\theta_{mm})\le T$, equivalently
% $r_{mm}\le r_{mm}(T,\theta_{mm})$ (since $D(\cdot,\theta_{mm})$ is strictly increasing).
%  On the set $\{r_{mm}:D(r_{mm},\theta_{mm})\le T\}$ there is no inventory penalty, and conditional on any $\theta_{rm}$, BD$_{mm}$'s payoff is proportional to
% \[(r_{rm}(T,\theta_{rm})-r_{mm})\,D(r_{mm},\theta_{mm}).\]
% By strict concavity, this function is single-peaked in $r_{mm}$ with peak at
% $r_{mm}^{A}(r_{rm}(T,\theta_{rm}))$. Condition \eqref{eq:BR-mm-condition} ensures that for every possible $\theta_{rm}$ the boundary point $r_{mm}(T,\theta_{mm})$ lies weakly to the left of the peak, so the objective is weakly increasing on $(-\infty,r_{mm}(T,\theta_{mm})]$. Hence among all feasible $r_{mm}\le r_{mm}(T,\theta_{mm})$, the maximizer is attained at the boundary $r_{mm}(T,\theta_{mm})$. 

% Thus BD$_{mm}$'s best response is $\sigma^T_{mm}(\theta_{mm})=r_{mm}(T,\theta_{mm})$.
% \end{proof}

\subsection{Existence and multiplicity of Bayesian equilibria}
\label{app:existence_mult_assymetric}

We now connect the above best-response conditions to the complete-information multiplicity established in Section \ref{sec:Multiple Equilibirum} of the main text (and Theorem~4.2 under perfect knowledge).

\begin{theorem****}[Existence of asymmetric-information equilibrium strategies ]
Maintain Assumptions~\ref{ass:regularity} and \ref{ass:multiplicative}. Suppose that at the minimal-type profile $(\underline\theta_{mm},\underline\theta_{rm})$ the complete-information game admits a constrained equilibrium with balanced volume $T^\ast\in(0,T_{\max}]$. Then the $T^\ast$-targeting strategies $(\sigma^{T^\ast}_{mm},\sigma^{T^\ast}_{rm})$ defined in \eqref{eq:T-target-strats} form a Bayesian Nash equilibrium.
\end{theorem****}

\begin{proof}
Let $T^*\in(0,T_{\max}]$ be the balanced volume of a constrained complete-information equilibrium at the minimal-type profile
$(\underline\theta_{mm},\underline\theta_{rm})$, with associated balanced rates
$r_{mm}^*=r_{mm}(T^*,\underline\theta_{mm})$ and $r_{rm}^*=r_{rm}(T^*,\underline\theta_{rm})$. By constrained optimality, $r_{mm}^*\le r_{mm}^A(r_{rm}^*)$ and $r_{rm}^*\ge r_{rm}^A(r_{mm}^*)$.

By Lemma 2, $r_{mm}(T^*,\theta_{mm})$ is decreasing in $\theta_{mm}$ and $r_{rm}(T^*,\theta_{rm})$ is increasing in $\theta_{rm}$. By Proposition~3 (monotonicity of peaks), $r_{mm}^A(\cdot)$ is increasing and $r_{rm}^A(\cdot)$ is increasing, so the worst-case peak comparisons occur at the minimal types. Therefore the conditions of Lemma 4 hold for target $T^*$, implying that for every type, the best response to $\sigma^{T^*}$ is to also target $T^*$. Hence $(\sigma^{T^*}_{mm},\sigma^{T^*}_{rm})$ is a BNE.
\end{proof}

\begin{theorem*****}[Multiplicity of asymmetric-information equilibrium strategies ]\label{thm:C-multiple}
Maintain Assumptions~\ref{ass:regularity} and \ref{ass:multiplicative}. Suppose that at the minimal-type profile $(\underline\theta_{mm},\underline\theta_{rm})$ the complete-information game admits at least two distinct constrained equilibria with balanced volumes $T_1\neq T_2$ in $(0,T_{\max}]$. Then the asymmetric information game admits at least two distinct Bayesian Nash equilibria, namely the $T_1$-targeting and $T_2$-targeting equilibria.

In particular, if the minimal-type complete-information game admits a continuum of constrained equilibrium volumes (as in Theorem 4.2), then the asymmetric information game admits a continuum of Bayesian Nash equilibria indexed by those volumes.
\end{theorem*****}
\begin{proof}
Apply Theorem~3 to each constrained minimal-type equilibrium volume $T_k$ ($k=1,2$). This yields that $(\sigma^{T_k}_{mm},\sigma^{T_k}_{rm})$ is a Bayesian Nash equilibrium for each $k$.

If $T_1\ne T_2$, then for any fixed type $\theta_{mm}$,
$r_{mm}(T,\theta_{mm})=\bar D^{-1}(T/\theta_{mm})$ is strictly increasing in $T$ (since $\bar D^{-1}$ is increasing), so $\sigma^{T_1}_{mm}(\theta_{mm})\ne\sigma^{T_2}_{mm}(\theta_{mm})$; similarly for BD$_{rm}$. Hence the equilibria are distinct. The continuum statement follows identically when the set of constrained minimal-type equilibrium volumes is an interval.
\end{proof}

% \begin{proof}
% Apply Theorem~\ref{thm:C-existence} to each constrained minimal-type equilibrium volume $T_k$ ($k=1,2$). This yields that $(\sigma^{T_k}_{mm},\sigma^{T_k}_{rm})$ is a Bayesian Nash equilibrium for each $k$.

% If $T_1\neq T_2$, the induced strategy profiles differ because for any fixed type $\theta_{mm}$,
% $r_{mm}(T,\theta_{mm})$ is strictly increasing in $T$ (by strict monotonicity of $\bar D$), and likewise $r_{rm}(T,\theta_{rm})$ is strictly increasing in $T$ (because $\bar S^{-1}$ is decreasing and $T/\theta$ is increasing in $T$). Hence the two equilibria are distinct.

% The final statement (continuum) follows by the same argument applied to an interval of distinct constrained minimal-type equilibrium volumes.
% \end{proof}

% === END: app_asymmetric_eqmv2.tex ===

% === INLINED: app_truthful_reporting.tex ===
\section{Truthful Reporting of Client Schedules}
\label{app:truthful_reporting}

In this appendix we provide a rigorous proof that it is a dominant strategy for  $BD$'s to truthfully report their client schedules in the smart contract protocol of Section  \ref{subsec:smart_contract_protocol}. 

\begin{proposition*****}[Truthful reporting of client schedules]\label{prop:truthful_schedules}
Fix the protocol step in which broker-dealers simultaneously report client schedules to the smart contract
(Step~2 of the protocol). Let the feasible trade volumes lie on a fixed grid $\mathcal{T}=\{T_1,\dots,T_I\}$.

For each $T\in\mathcal{T}$, let $r^{c}_{mm}(T)$ denote the $MM$ client's reservation schedule (the minimum repo rate at which $MM$ is willing to lend at volume $T$), and let $r^{c}_{rm}(T)$ denote the $RM$ client's reservation schedule (the maximum repo rate at which $RM$ is willing to borrow at volume $T$).

Assume:
\begin{enumerate}
\item \textbf{Client authorization / commitment.} A reported schedule must be accompanied by a client signature (or other binding commitment, such as deposit) that obligates the client--dealer pair to execute the client leg
at the reported rate for the volume selected by the smart contract.
\item \textbf{No side agreements (no collusion).} Client--dealer pairs execute at the reported rates; in particular, they do not use off-protocol side payments or renegotiation to implement terms different from
those reported to the smart contract.

\item \textbf{Client rationality constraint.} MM will not authorize a schedule that offers it a rate below its reservation schedule, i.e.\ any authorized report $\hat r_{mm}(\cdot)$ satisfies
$\hat r_{mm}(T)\ge r^{c}_{mm}(T)$ for all $T\in\mathcal{T}$.
Similarly, RM will not authorize a schedule that charges it a rate above its reservation schedule, i.e.\
any authorized report $\hat r_{rm}(\cdot)$ satisfies
$\hat r_{rm}(T)\le r^{c}_{rm}(T)$ for all $T\in\mathcal{T}$.
\item \textbf{Smart-contract selection rule.} Given reported schedules $\hat r_{mm}(\cdot)$ and
$\hat r_{rm}(\cdot)$ and any additional exogenous admissibility constraints (e.g.\ minimum-spread hurdles),
the smart contract selects a trade volume $T^\ast\in\mathcal{T}$ that maximizes joint intermediation profit
\[\Pi(T)\;:=\;\bigl(\hat r_{rm}(T)-\hat r_{mm}(T)\bigr)\,T\]
over the admissible set (and aborts with payoff $0$ if no admissible $T$ exists).
\item \textbf{Surplus split.} Each dealer's payoff from the selected trade is a fixed share of $\Pi(T^\ast)$
(e.g.\ $\frac12\Pi(T^\ast)$ under equal split).
\end{enumerate}
Then it is a dominant strategy for BD$_{mm}$ to report $\hat r_{mm}(T)=r^{c}_{mm}(T)$ for all
$T\in\mathcal{T}$, and a dominant strategy for BD$_{rm}$ to report $\hat r_{rm}(T)=r^{c}_{rm}(T)$ for all
$T\in\mathcal{T}$.
\end{proposition*****}

\begin{proof}
We prove the claim for BD$_{mm}$; the argument for BD$_{rm}$ is symmetric.

\medskip
\noindent\textbf{BD$_{mm}$.}
Fix an arbitrary report $\hat r_{rm}(\cdot)$ by BD$_{rm}$ and fix any exogenous admissibility constraints
(including the spread constraint, if any). Consider any authorized report $\hat r_{mm}(\cdot)$ by
BD$_{mm}$. By the client rationality constraint, for every $T\in\mathcal{T}$ we have
$\hat r_{mm}(T)\ge r^{c}_{mm}(T)$.

For any $T\in\mathcal T$, $ \hat r_{mm}(T)\ge r^c_{mm}(T)$ implies
\[(\hat r_{rm}(T)-\hat r_{mm}(T))T \;\le\; (\hat r_{rm}(T)-r^c_{mm}(T))T.\]
Moreover, any admissibility constraint that is weakly harder to satisfy when $\hat r_{mm}(T)$ increases (e.g.\ a minimum-spread filter) yields the set inclusion
$A(\hat r_{mm},\hat r_{rm})\subseteq A(r^c_{mm},\hat r_{rm})$. Therefore the maximum attainable joint surplus under $\hat r_{mm}$ is weakly below that under $r^c_{mm}$, so BD$_{mm}$ cannot
profitably deviate from truthful reporting.

\medskip
\noindent\textbf{BD$_{rm}$ (symmetric).}
Fix $\hat r_{mm}(\cdot)$. Any authorized report $\hat r_{rm}(\cdot)$ must satisfy $\hat r_{rm}(T)\le r^{c}_{rm}(T)$ for all $T$.
Lowering $\hat r_{rm}(T)$ weakly reduces $\bigl(\hat r_{rm}(T)-\hat r_{mm}(T)\bigr)T$ pointwise and can only
shrink the set of $T$ satisfying any minimum-spread constraint. Hence the same set-inclusion and pointwise-comparison argument implies BD$_{rm}$ weakly maximizes its payoff by reporting
$\hat r_{rm}(T)=r^{c}_{rm}(T)$ for all $T$, i.e.\ truthful reporting is a dominant strategy for BD$_{rm}$.
\end{proof}

% === END: app_truthful_reporting.tex ===

% === INLINED: App_crypto_implementations.tex ===
\section{Details for Smart Contract Implementation}
\label{app:crypto_implementation} 

This appendix discusses the cryptography required to implement the smart contract functions on a stand alone server and a blockchain. 

\paragraph{Implementation A: \smc inside a TEE (hardware‑backed privacy, ZK auditability).}
The comparison/optimization code that parses $S_{r_{rm}},D_{r_{mm}}$, computes equation \ref{eq:sc_objective}, selects $\{r_{rm,*},r_{mm,*},T^{*}\}$,
and runs the proof generator is loaded into a TEE. Remote attestation binds a code hash (and the verifier/proving keys) to the enclave. After verifying attestation, each $BD$ transmits its full schedule over the attested secure channel (or sends schedule elements sufficient for membership checking). Inside the enclave, the \smc computes $T^{*}=\arg\max_{T} \pi(T)$ over the admissible set (e.g., intersection/feasible pairs induced by the two schedules and hurdle‑rate filter), then outputs: (a) $E_{rm},E_{mm}$, the per \bd encryptions of $T^{*}$, and (b) a zk‑SNARK/NIZK $\pi$ proving membership and optimality of $\{r_{rm,*},r_{mm,*},T^{*}\}$relative to $C_{rm},C_{mm}$ and the attested program. Each $BD$ verifies the attestation and $\pi$; if either check fails or a deadline expires, the protocol aborts. Privacy rests on enclave isolation; auditability rests on the public verification of $\pi$ (not on trust in the operator).\footnote{ Costan and Devadas \cite{costan2016sgx}, Bauman et.al. \cite{baumann2014haven} and  Schuster et.al \cite{schuster2015vc3} are descriptions of TEE implementations.}

\paragraph{Implementation B: \smc as a public‑chain contract (trustless verification, on‑chain audit).}
The contract records $C_{rm},C_{mm}$ and enforces the phase timeouts. Off‑chain, a prover computes $\{r_{rm,*},r_{mm,*},T^{*}\}$from the private schedules and generates a zk‑SNARK/NIZK $\pi$ for the statement above (with schedule‑size bounds fixed ex ante). On‑chain, the contract verifies $\pi$ against $(C_{rm},C_{mm})$ and, if valid before the deadline, publishes only $E_{rm}= Enc_{pk_{rm}}(T^{*})$ and $E_{mm}=ENC_{pk_{mm}}(T^{*})$. The plaintext $\{r_{rm,*},r_{mm,*},T^{*}\}$never appears on‑chain; each $BD$ decrypts locally. If commitments or proof are not posted by their respective deadlines, the contract cancels without leakage. Anyone can re‑verify $\pi$, yielding public auditability.

The off-chain phase runs either as a two-party computation between $BD_{rm}$ and $BD_{mm}$ on their own machines or via a small set of independent multiparty computation (``MPC'') servers. In all cases, each $BD$’s schedule is secret-shared (or garbled) so that no single processor ever observes the other’s plaintext input. The prover jointly computes $\{r_{rm,*}, r_{mm,*},T^{*}\}$ and produces a zk-proof $\pi$ that (i) the private inputs used open the on-chain commitments $(C_{rm},C_{mm})$, (ii) $\{r_{rm,*}, r_{mm,*},T^{*}\}$ is admissible and maximizes the objective (equation \ref{eq:sc_objective}) over the bounded grid, and (iii) the posted ciphertexts $CT_{rm},CT_{mm}$ are valid encryptions of the same trades. Thus privacy derives from MPC (simulation-based), integrity from input-binding to commitments, and auditability from on-chain verification of $\pi$; only the encrypted $\{r_{rm,*}, r_{mm,*},T^{*}\}$ is delivered to each $BD$, and missed deadlines trigger cancellation.

\paragraph{Remarks.}
Circuit size and proving cost scale with schedule bounds and the complexity of $\i(\cdot)$; using vector commitments (Merkle/polynomial) and bounded grids $\mathbb{T}$ renders the “max” check implementable via batched pairwise comparisons inside the circuit. The approach is agnostic to the concrete proof system (zk‑SNARK or NIZK) and to platform: any TEE with remote attestation or any chain with a verifier suffices. In both implementations, $BD_{rm}$ and $BD_{mm}$ learn only the same $\{r_{rm,*},r_{mm,*},T^{*}\}$, sent privately by the \smc, and can independently verify that it is exactly the trade the mechanism should select under equation \ref{eq:JPMconstrained}.
% === END: App_crypto_implementations.tex ===

\end{document}